\documentclass[review,hidelinks,onefignum,onetabnum]{siamart251216}

\usepackage{amsfonts}
\usepackage{amssymb}
\usepackage{amstext}
\usepackage{mathrsfs}
\usepackage{mathtools}
\usepackage{graphicx}

\newsiamremark{remark}{Remark}
\newsiamremark{example}{Example}
\newsiamthm{assumption}{Assumption}
\crefname{assumption}{Assumption}{Assumptions}

\headers{Approximate Nash Equilibria in General Markov Games with Borel Spaces}{N. Saldi, G. Arslan, and S. Y\"{u}ksel}

\title{General-Sum Markov Games with Borel Spaces: Existence of Approximate Nash Equilibria via Finite Model Approximations under Setwise or Weak Continuity\thanks{Submitted to the editors DATE.}}

\author{Naci Saldi\thanks{Department of Mathematics, Bilkent University, Cankaya, Ankara, Turkey (\email{naci.saldi@bilkent.edu.tr}).}
\and G\"{u}rdal Arslan\thanks{Department of Electrical Engineering, University of Hawaii at Manoa, Honolulu, HI, USA (\email{gurdal@hawaii.edu}).}
\and Serdar Y\"{u}ksel\thanks{Department of Mathematics and Statistics, Queen's University, Kingston, Ontario, Canada (\email{yuksel@queensu.ca}).}}

\ifpdf
\hypersetup{
  pdftitle={General-Sum Markov Games with Borel Spaces: Existence of Approximate Nash Equilibria via Finite Model Approximations under Setwise or Weak Continuity},
  pdfauthor={N. Saldi, G. Arslan, and S. Yuksel}
}
\fi

\def\rF{\mathbb{F}}
\def\R{\mathbb{R}}
\def\N{\mathbb{N}}

\def\diam{\mathop{\rm diam}}

\def\intr{\mathop{\rm int}}

\def\argmin{\mathop{\rm arg\, min}}
\def\cE{\mathbb{E}}

\def\B{{\mathcal B}}

\def\P{{\mathcal P}}

\def\S{{\mathsf S}}

\def\sX{{\mathsf X}}

\def\sA{{\mathsf A}}
\def\sH{{\mathsf H}}

\def\sM{{\mathsf M}}
\def\sR{{\mathsf R}}
\def\sE{{\mathsf E}}

\renewcommand{\P}{\mathcal P}          

\allowdisplaybreaks

\nolinenumbers

\begin{document}

\maketitle

\begin{abstract}
Establishing the existence of exact or near Markov or stationary perfect Nash equilibria in nonzero-sum Markov games over Borel spaces is a challenging problem with limited positive results. Motivated by problems in multi-agent and Bayesian learning, this paper demonstrates the existence of approximate Markov and stationary Nash equilibria for such games under relatively mild regularity conditions; notably under setwise continuity. Our approach is constructive: For both compact and non-compact state spaces, we approximate the Borel model with finite state-action models and show that their equilibria correspond to \(\epsilon\)-equilibria for the original game. Compared with previous results in the literature, which we comprehensively review, we provide more general conditions, along with explicit approximation models whose equilibria are $\epsilon$-equilibria for the original model. For completeness, we also study the approximation of zero-sum Markov games and Markov teams to highlight the key differences between zero-sum and nonzero-sum settings. In particular, while for zero-sum and team games, joint weak (Feller) continuity of the transition kernel is sufficient (as the value function is continuous), this is not the case for general nonzero-sum games.
\end{abstract}



\section{Introduction}

Shapley \cite{Sha53} established the existence of stationary equilibria for discounted zero-sum stochastic (or Markov) games. Fink subsequently extended Shapley's framework to the nonzero-sum case \cite{Fin64}, proving the existence of stationary equilibria for discounted games with finite state and action spaces. His proof relies on the Kakutani-Fan-Glicksberg fixed point theorem (see \cite[Corollary 17.55]{AlBo06}) applied directly to the best-response correspondence on the space of stationary policies. This result is particularly important because it allows for the analysis of games in which players' interests are not strictly opposed.

Fink's policy-space approach, however, does not extend readily to models with uncountable state spaces. The main difficulty is the lack of a suitable topology on the space of policies under which the best-response correspondence satisfies the conditions of the Kakutani-Fan-Glicksberg theorem, in particular the closed graph property. In general, the value functions induced by a sequence of policies need not converge continuously under natural topologies on the policy space. Consequently, the closed graph property of the best-response correspondence cannot be established without imposing strong additional assumptions on the model.

An alternative approach to establishing stationary or Markov perfect Nash equilibria in nonzero-sum games is to work directly in the space of value functions through an auxiliary ``one-shot'' game with absorbing costs \cite{mertens2015repeated}. The argument proceeds in two steps. First, if there exists a vector-valued function \(f\) on the state space such that, for every initial state \(x\), the one-shot game with absorbing cost \(f\) admits an equilibrium whose expected cost is \(f(x)\), then \(f\) represents an equilibrium cost in the discounted game. This characterization is closely related to the Bellman optimality equation for discounted Markov decision processes (MDPs) and to the corresponding min-max optimality equation for discounted zero-sum Markov games. Second, one seeks such a function \(f\) as a fixed point of the associated correspondence on the value-function space.

For discounted MDPs and zero-sum discounted Markov games, Banach's fixed point theorem is sufficient to establish the existence of the corresponding value functions. In the nonzero-sum case, however, the analysis is substantially more involved and typically requires the Kakutani-Fan-Glicksberg theorem. When the state space is uncountable, applying this theorem imposes stringent requirements on the model components, such as norm continuity of the transition kernel with respect to the action variables and restrictive separability or state-independence conditions on the transition kernel and cost function. Moreover, the correspondence of best-response value functions generally lacks convexity, while the resulting value functions are typically only measurable rather than continuous. These issues make the nonzero-sum setting considerably more intricate than the zero-sum case. We refer the reader to \cite[Section 7]{JaNo24} for a detailed comparison of these approaches.

To address the lack of convexity issue, some further relaxations are often imposed: (i) introducing a correlation device, as in \cite{NoRa92}, which allows for the establishment of correlated stationary equilibria, or (ii) imposing restrictive separability or state independence conditions on the transition kernels and cost functions and requiring that the transition law is absolutely continuous with respect to some fixed nonatomic probability distribution \cite{himmelberg1976existence} \cite{parthasarathy1989existence} \cite{MePa91} and thus avoiding the usage of any additional external public random signal or (iii) imposing the usage of memory \cite{JaNo16}.

Due to the difficulties outlined above, research on the existence of Markov or stationary equilibria for non-zero-sum discounted Markov games with standard Borel spaces has been fragmented. \cite{levy2013discounted,levy2015corrigendum} reported the non-existence of stationary equilibrium policies for such games with continuous spaces. Notably, there have been few positive results regarding the existence of stationary equilibria in discounted Markov games with continuous state spaces. For example, \cite{himmelberg1976existence} imposed restrictive separability conditions on the transition kernels and cost functions, \cite{parthasarathy1989existence} required state independence of transitions, and \cite{JaNo16} presented conditions for equilibrium under policies with memory. We refer the reader to \cite[Chapter 6]{BaZa18}, \cite{JaNo18-b} for a comprehensive discussion on this topic. Nevertheless, despite the general challenges in proving the existence of Markov or stationary equilibria in nonzero-sum games, the existence of history-dependent subgame perfect equilibria can still be established under mild conditions \cite{Sol98,MePa03}. 

The situation in the finite horizon case allows for more positive results than in the infinite horizon discounted cost scenario. A seminal study on this topic is conducted by Rieder \cite{rieder1979equilibrium}, who examines finite-horizon Markov games with standard Borel spaces. His approach employs a direct method based on backward induction, combined with an ingenious measurable selection argument. This analysis is specifically tailored for finite horizons. In Rieder's work, the existence of Markov policies is established, which is a common consideration in finite-horizon settings, under the setwise continuity condition on the transition kernel with respect to the action variable.

As is common in the literature, if an exact solution cannot be found or its existence cannot be established, one seeks approximate solutions. To this end, in this paper, we investigate the existence of approximate Nash equilibria obtained through finite state-action approximations under the most general conditions known to us. Our primary motivation for addressing this problem stems from its applications in multi-agent learning algorithms, as explicitly discussed in \cite{yongacoglu2023satisficing,yongacoglu2024independent}. These studies highlight the convergence to $\varepsilon$-equilibrium policies through policy revision processes along $\varepsilon$-satisficing paths \cite{yongacoglu2024independent} (where an agent revises a policy only when they are not $\varepsilon$-satisfied). Specifically, the existence of $\varepsilon$-equilibria is a sufficient condition for ensuring the convergence of the independent learning algorithms developed in these studies for a large class of stage games.

On the existence of $\varepsilon$-equilibria for games with uncountable state and action spaces, Whitt \cite{whitt1980representation} studies approximation schemes under assumptions that are stronger than ours, most notably a uniform form of total variation convergence of the transition probabilities. Moreover, the approximating models are obtained through an existence argument rather than being explicitly constructed.

Nowak \cite{nowak1985existence}, on the other hand, uses a different approximation approach based on a density representation of the transition kernel. In particular, the state space is assumed to be a countably generated measurable space, and the transition kernel is required to admit a density with respect to a fixed reference measure. This density is viewed as an element of a function-valued $L_1$-space, and the separability of this space is used to establish the existence of an approximate model with countably many states. Thus, Nowak's construction relies essentially on the density assumption and the resulting $L_1$-space structure, and the approximating model is not explicitly constructed. In contrast, our approach constructs the approximating model explicitly through a finite-state quantization of the state space. Consequently, our approximating models have finitely many states, rather than merely countably many states.

Our assumption on the transition kernel is substantially weaker than the regularity conditions typically imposed in finite-model approximation results. In particular, we require only setwise continuity of the transition probability, rather than the stronger total variation continuity. Thus, our approximation result remains valid under substantially weaker continuity requirements. Additionally, we impose no density assumption on the transition kernel. Consequently, our result applies to a strictly broader class of transition kernels, including kernels that may fail to admit a density and may not satisfy the total variation continuity condition required in \cite{nowak1985existence}. This demonstrates that our approximation framework extends beyond the class of models covered by the existing assumptions in \cite{nowak1985existence}.

Building on similar ideas as in \cite{nowak1985existence}, Nowak and Altman \cite{NoAl02} address the same approximation problem for unbounded one-stage costs under discounted and average cost criteria. For the average cost criterion, they assume geometric ergodicity type conditions as it is common for the average cost criterion. In our approach, (i) the approximation is constructive and explicit; (ii) a density assumption is not required, and (iii) compactness is not required. 


Our constructive approximation result and method are achieved via an explicit finite-game construction, building on recent work \cite{SaYuLi15c,SaLiYuSpringer} that developed finite approximations for MDPs with standard Borel spaces. These methods were initially applied under the assumption of weak continuity for the transition probabilities and were shown to yield near-optimal approximations. In our current study, given the game-theoretic nature of the problem, we impose setwise continuity condition on the kernel in place of weak continuity. This is because, unlike standard weakly continuous MDPs, where value functions are continuous in the state, value functions in nonzero-sum games generally lack this regularity. Indeed, as Rieder \cite{rieder1979equilibrium} suggests, in such cases, we may need to settle for merely measurable value functions. Nonetheless, our conditions are still significantly less restrictive than total variation continuity. Additionally, we extend our analysis to encompass non-compact state spaces. 

\smallskip

{\bf Contributions of the Paper.} (i) We establish conditions for the existence of near Markov and stationary Nash equilibria in nonzero-sum stochastic games with Borel spaces, under both finite-horizon and discounted cost criteria. Our results apply to both compact and non-compact state spaces and, most importantly, require substantially weaker assumptions than those typically imposed in the literature. In particular, we relax the restrictive total variation continuity of the transition kernel to the weaker setwise continuity condition and do not impose any density assumption on the transition probabilities (see Theorem~\ref{finite_horizon_thm_2} and Theorem~\ref{discounted_thm_2}). This allows us to establish near equilibrium existence under a significantly broader class of stochastic games than considered in existing results. (ii) Furthermore, we establish the existence of near Markov and stationary Nash equilibria through a finite state-action model approximation of the original model. Our finite model construction is explicit, based on the method developed in \cite{SaYuLi15c,SaLiYuSpringer} for obtaining finite approximations of MDPs with Borel spaces. Specifically, we construct a finite Markov game model and demonstrate that, for every $\varepsilon > 0$, a sufficiently fine approximation of the original model exists. This approximation ensures the existence of an equilibrium for the finite model (as guaranteed by \cite{Fin64,rieder1979equilibrium}), which serves as an approximate Nash equilibrium for the original problem. Thus, our result not only establishes existence but also provides an explicit method to compute or learn near Markov quilibria. This approach has implications for multi-agent learning problems involving general spaces and information structures beyond finite models \cite{yongacoglu2024independent}. (iii) By demonstrating continuous convergence as the approximation becomes finer, our contribution also establishes a positive result on the continuous dependence of equilibria in the refinement of information structures (that is, as the quantization of information gets more refined) within Markov games-a question for which there are few positive results \cite{hogeboom2023continuity}. (iv) We also explore the approximation of zero-sum discounted Markov games and Markov teams. Establishing the existence of stationary equilibria in zero-sum setting is considerably simpler than in nonzero-sum setting (due to continuity properties of the value function in zero-sum and team games, unlike the nonzero-sum setting). Therefore, this section serves to highlight the key differences between zero-sum and nonzero-sum settings, emphasizing why the latter presents a more challenging analysis. 

{\bf Notation.} For a metric space $\sE$, the Borel $\sigma$-algebra (the smallest $\sigma$-algebra that contains the open sets of $\sE$) is denoted by $\B(\sE)$.  We let $B(\sE)$ and $C_b(\sE)$ denote the set of all bounded Borel measurable and bounded continuous real functions on $\sE$, respectively. For any $u \in C_b(\sE)$ or $ u \in B(\sE)$, let $\|u\| \coloneqq \sup_{e \in \sE} |u(e)|$
which turns $C_{b}(\sE)$ and $B(\sE)$ into Banach spaces. Let $\P(\sE)$ denote the set of all probability measures on $\sE$. A sequence $\{\mu_n\}$ of probability measures on $\sE$ is said to converge weakly (resp., setwise) (see \cite{HeLa03}) to a probability measure $\mu$ if $\int_{\sE} g(e) \mu_n(de)\rightarrow \int_{\sE} g(e) \mu(de) \text{ for all } g\in C_b(\sE)$ (resp., for all $g \in B(\sE)$). For any $\mu,\nu \in \P(\sE)$, the total variation distance between $\mu$ and $\nu$, denoted as $\|\mu - \nu\|_{TV}$, is equivalently defined as
\begin{align}
	\|\mu-\nu\|_{TV} &\coloneqq 2 \sup_{D \in \B(\sE)} |\mu(D) - \nu(D)| = \sup_{\|g\| \leq 1} \biggl| \int_{\sE} g(e) \mu(de) - \int_{\sE} g(e) \nu(de) \biggr|. \nonumber
\end{align}
Unless otherwise specified, the term `measurable' will refer to Borel measurability.

\section{Nonzero Sum Stochastic Games}\label{game_model}

A discrete-time nonzero sum stochastic game can be described by a tuple
\begin{align}
	\bigl( \sX, \sA^1,\ldots, \sA^N, c^1,\ldots, c^N, p \bigr), \nonumber
\end{align}
where Borel spaces (i.e., Borel subsets of complete and separable metric spaces) $\sX$ and $\{\sA^i\}_{i=1}^N$ denote the \emph{state} and \emph{action} spaces, respectively. The \emph{stochastic kernel} 
$$
p: \sX \times {\boldsymbol \sA} \ni (x,{\bf a}) \mapsto p(\cdot|x,{\bf a}) \in \P(\sX)
$$
denotes the \emph{transition probability} of the next state given that previous state and actions are $(x,\bf a)$ (see \cite{HeLa96}), where 
$$
{\boldsymbol \sA} \coloneqq \prod_{i=1}^N \sA^i, \,\,\,\, {\bf a} = (a^1,\ldots,a^N).
$$
Hence, it satisfies: (i) $p(\,\cdot\,|x,\bf a)$ is an element of $\P(\sX)$ for all $(x,\bf a) \in \sX \times \boldsymbol \sA$, and (ii) $p(D|\,\cdot\,,\,\cdot\,)$ is a measurable function from $\sX\times{\boldsymbol \sA}$ to $[0,1]$ for each $D\in\B(\sX)$. The \emph{one-stage cost} function $c^i$ for player~$i$ is a measurable function from $\sX \times {\boldsymbol \sA}$ to $\R$.

Define the history spaces 
$$
\sH_0 = \sX, \,\,\,\, \sH_{t}=(\sX\times{\boldsymbol \sA})^{t}\times\sX, \,\, t \geq 1 
$$
endowed with their product Borel $\sigma$-algebras generated by $\B(\sX)$ and $\B(\sA^i)$, $i=1,\ldots,N$. A \emph{policy} for player~$i$ is a sequence $\pi^i=\{\pi_{t}^i\}$ of stochastic kernels on $\sA^i$ given $\sH_{t}$. The set of all policies for player~$i$ is denoted by $\Pi^i$. Let $\Phi^i$ denote the set of stochastic kernels on $\sA^i$ given $\sX$, and let $\rF^i$ denote the set of all measurable functions from $\sX$ to $\sA^i$. A \emph{randomized Markov} policy for player~$i$ is a sequence $\pi^i=\{\pi_{t}^i\}$ of stochastic kernels on $\sA^i$ given $\sX$. A \emph{deterministic Markov} policy is a sequence of stochastic kernels $\pi^i=\{\pi_{t}^i\}$ on $\sA^i$ given $\sX$ such that $\pi_{t}^i(\,\cdot\,|x)=\delta_{f_t(x)}(\,\cdot\,)$ for some $f_t \in \rF^i$, where $\delta_z$ denotes the point mass at $z$. The set of randomized and deterministic Markov policies for player~$i$ are denoted by $\sR\sM^i$ and $\sM^i$, respectively. A \emph{randomized stationary} policy for player~$i$ is a
constant sequence $\pi^i=\{\pi_{t}^i\}$ of stochastic kernels on $\sA^i$ given $\sX$ such that
$\pi_{t}^i(\,\cdot\,|x)=\varphi(\,\cdot\,|x)$ for all $t$ for some
$\varphi \in \Phi^i$. A \emph{deterministic stationary} policy is a constant sequence of stochastic kernels $\pi^i=\{\pi_{t}^i\}$ on $\sA^i$ given $\sX$ such that $\pi_{t}^i(\,\cdot\,|x)=\delta_{f(x)}(\,\cdot\,)$ for all $t$ for some
$f \in \rF^i$. The set of randomized and deterministic stationary policies for player~$i$ are identified with the sets $\Phi^i$ and $\rF^i$, respectively. Hence, we have 
$$
\rF^i \subset \sM^i \subset \Pi^i, \,\,\, \Phi^i \subset \sR\sM^i \subset \Pi^i, \,\,\, \rF^i \subset \Phi^i, \,\,\, \sM^i \subset \sR\sM^i.
$$
According to the Ionescu Tulcea theorem (see \cite{HeLa96}), an initial distribution $\mu$ on $\sX$ and a joint policy ${\boldsymbol \pi} \coloneqq (\pi^1,\ldots,\pi^N)$ define a unique probability measure $P_{\mu}^{\boldsymbol \pi}$ on $\sH_{\infty}=(\sX\times{\boldsymbol \sA})^{\infty}$.
The expectation with respect to $P_{\mu}^{\boldsymbol \pi}$ is denoted by $\cE_{\mu}^{\boldsymbol \pi}$. If $\mu=\delta_x$, we write $P_{x}^{\boldsymbol \pi}$ and $\cE_{x}^{\boldsymbol \pi}$ instead of $P_{\delta_x}^{\boldsymbol \pi}$ and $\cE_{\delta_x}^{\boldsymbol \pi}$. For player~$i$, the cost functions to be minimized in this paper are the finite-horizon cost and the $\beta$-discounted cost, respectively given by
\begin{align}
	J^i(\boldsymbol \pi,x) &= \cE_{x}^{\boldsymbol \pi}\biggl[\sum_{t=0}^{T-1} c^i(x_{t},{\boldsymbol a}_{t})\biggr], \nonumber \\
	J^i(\boldsymbol \pi,x) &= \cE_{x}^{\boldsymbol \pi}\biggl[\sum_{t=0}^{\infty}\beta^{t} \, c^i(x_{t},{\boldsymbol a}_{t})\biggr]. \nonumber
\end{align}

\begin{remark}
	We observe that the infinite sum $\sum_{t=0}^{\infty} \beta^t c^i(x_t, a_t)$ may not be finite or well-defined in the definition of $J^i$ if $c^i$ is assumed only to be measurable. However, additional assumptions introduced in subsequent sections guarantee the well-defined nature of $J^i$.
\end{remark}

In the definition below, the cost is either finite-horizon or discounted. 

\begin{definition}[Nash equilibrium]
	A joint policy $\boldsymbol \pi^{*}$ is said to be Nash $\varepsilon$-equilibrium ($\varepsilon \geq 0$) for the initial state $x$ if 
	$$
	J^i(\boldsymbol \pi^{*},x) \leq  \inf_{\pi^i \in \Pi^i} J^i(\boldsymbol \pi^{-i},\pi^i,x) + \varepsilon,
	$$
	for all $i=1,\ldots,N$, where $\boldsymbol \pi^{-i} \coloneqq \boldsymbol \pi \setminus \{\pi^i\}$. If $\varepsilon = 0$, it is called a Nash equilibrium. 
\end{definition}

In nonzero-sum stochastic games, the primary objective is to establish the existence of (and, if possible, compute or learn) a Markov Nash equilibrium for the finite-horizon case and a stationary Nash equilibrium for the discounted case. As noted, this task becomes especially challenging when dealing with uncountable Borel spaces in the state and action spaces. In contrast, for finite cases, where both the state and action spaces are finite, the existence of such equilibria can be established relatively easily, as shown in \cite{Fin64,rieder1979equilibrium}. To address the challenge of establishing approximate Markov or stationary Nash equilibria in uncountable cases, we employ a finite approximation method. By approximating our model with a finite one, we prove that the stationary or Markov Nash equilibrium of the finite model serves as an approximate equilibrium for the original problem. 

\section{$\delta$-Approximation of $N$-player Game} \label{finite_model}

In this section, we construct the finite approximation of the game model introduced in the previous section. We impose the assumptions below on the components of the original nonzero sum stochastic game.

	\begin{assumption}
		\label{compact:as1}
		\begin{itemize}
			\item [  ]
			\item [(a)] The one-stage cost functions $c^i$ are in $C_b(\sX \times \boldsymbol \sA)$.
			\item [(b)] The stochastic kernel $p(\,\cdot\,|x, \boldsymbol a)$ is setwise continuous in $(x,\boldsymbol a)$.
			\item [(c)] $\sX$ and $\sA^i$, $i=1,\ldots,N$, are compact.
		\end{itemize}
	\end{assumption}

Let $d_{\sX}$ and $d_{\sA^i}$ denote the metric on $\sX$ and $\sA^i$, respectively, for all $i=1,\ldots,N$. Fix any $\delta > 0$. Since the state space $\sX$ and action spaces $\sA^i$ are assumed to be compact, one can find 
finite sets $\sX_{\delta} = \{x_1,\ldots,x_{k_{\delta}}\}$ and $\sA_{\delta}^i = \{a_1^i,\ldots,a_{h_{\delta}}^i\}$ such that they are $\delta$-nets in the corresponding uncountable spaces. Define functions $Q_{\delta}$ and $Q_{i,\delta}$ by
\begin{align}
	Q_{\delta}(x) &\coloneqq \argmin_{z \in \sX_{\delta}} d_{\sX}(x,z),\nonumber \\
	Q_{i,\delta}(a^i) &\coloneqq \argmin_{b^i \in \sA_{\delta}^i} d_{\sA^i}(a^i,b^i), \nonumber 
\end{align}
where ties are broken so that functions are measurable. Note that $Q_{\delta}$ induces a partition $\{\S_{\delta}^i\}_{i=1}^{k_{\delta}}$ on the state space $\sX$ given by
\begin{align}
	\S_{\delta}^i = \{x \in \sX: Q_{\delta}(x) = x_{i}\}, \nonumber
\end{align}
with diameter $\diam(\S_{\delta}^i) \leq 2 \, \delta$. Let $\nu_{\delta}$ be probability measures on $\sX$ satisfying
\begin{align}
	\nu_{\delta}(\S_{\delta}^i) > 0 \text{  for all  } i=1,\ldots,k_{\delta}.  \nonumber
\end{align}
We let $\nu_{\delta}^i$ be the restriction of $\nu_{\delta}$ to $\S_{\delta}^i$ defined by
\begin{align}
	\nu_{\delta}^i(\,\cdot\,) \coloneqq \frac{\nu_{\delta}(\,\cdot\,)}{\nu_{\delta}(\S_{\delta}^i)}. \nonumber
\end{align}
The measures $\nu_{\delta}^i$ will be used to define a finite game model with resolution $\delta$. To this end, the one-stage cost functions $c_{\delta}^i: \sX_{\delta}\times \boldsymbol \sA \rightarrow \R$ and the transition probability $p_{\delta}$ on $\sX_{\delta}$ given $\sX_{\delta}\times \boldsymbol \sA$ are defined by
\begin{align}
	c_{\delta}^i(x_i,\boldsymbol a) &\coloneqq \int_{\S_{\delta}^i} c^i(x,\boldsymbol a) \, \nu_{\delta}^i(dx), \nonumber \\ \\
	p_{\delta}(\,\cdot\,|x_i,\boldsymbol a) &\coloneqq \int_{\S_{\delta}^i} Q_{\delta} \ast p(\,\cdot\,|x, \boldsymbol a) \,  \nu_{\delta}^i(dx), \nonumber
\end{align}
where $Q_{\delta}\ast p(\,\cdot\,|x, \boldsymbol a) \in \P(\sX_{\delta})$ is the pushforward of the measure $p(\,\cdot\,|x,\boldsymbol a)$ with respect to $Q_{\delta}$. For each $\delta$, we define $\delta$-approximation of the original game as a finite nonzero sum stochastic game with the following components: $\sX_{\delta}$ is the state space, $\{\sA^i_{\delta}\}_{i=1}^N$ are the action spaces, $p_{\delta}$ is the transition probability and $\{c_{\delta}^i\}_{i=1}^N$ are the one-stage cost functions. History spaces, policies and cost functions are defined in a similar way as in the original model. To distinguish them from the original game model, we add $\delta$ as a subscript in each object for the finite model.

\subsection{Some Fundamental Results}\label{fundamental}

In this section, we prove some fundamental results that will be used repeatedly throughout the paper. To this end, we start with the following one. Let $B(\sX; L)$ denote the set of functions in $B(\sX)$ that are uniformly bounded by the constant $L$ in sup-norm. 

\begin{lemma} \label{result1}
Suppose that Assumption~\ref{compact:as1} holds (in particular (b) and (c)). Let $f_n$ converge to $f$ pointwise, where $f_n,f \in  B(\sX; L)$ for some $L>0$, and let $x_n$ converge to $x$ in $X$. Then we have
$$
\lim_{n \rightarrow \infty} \sup_{\boldsymbol a \in \boldsymbol \sA} \left| \int_{\sX} f_n(y) \, p(dy|x_n,\boldsymbol a) - \int_{\sX} f(y) \, p(dy|x,\boldsymbol a) \right| = 0.
$$ 
Similarly, if $\boldsymbol a_n$ converges to $\boldsymbol a$ in $\boldsymbol \sA$, then we have
$$
\lim_{n \rightarrow \infty} \sup_{x \in \sX} \left| \int_{\sX} f_n(y) \, p(dy|x,\boldsymbol a_n) - \int_{\sX} f(y) \, p(dy|x,\boldsymbol a) \right| = 0.
$$ 
The same results also hold if we replace $f$ with $f_n$ in the second integrals. Finally, the same result holds when taking the supremum over $\sX \times \boldsymbol\sA^{-i}$ and considering convergence over $\sA^i$, for each $i=1,\ldots,N$.
\end{lemma}

\begin{proof}
We only prove the first claim as the proof of the other claims are similar.

For each $n$, define 
$$
l_n(\boldsymbol a) \coloneqq \left| \int_{\sX} f_n(y) \, p(dy|x_n,\boldsymbol a) - \int_{\sX} f(y) \, p(dy|x,\boldsymbol a) \right|.
$$
By Assumption~\ref{compact:as1}-(b), $l_n$ is (uniformly) continuous. Since $\boldsymbol \sA$ is compact, for each $n$, there exists $\boldsymbol a_n$ such that 
\begin{align*}
&\sup_{\boldsymbol a \in \boldsymbol \sA} \left| \int_{\sX} f_n(y) \, p(dy|x_n,\boldsymbol a) - \int_{\sX} f(y) \, p(dy|x,\boldsymbol a) \right| \\
&\phantom{xxxxxxxxxx} = \left| \int_{\sX} f_n(y) \, p(dy|x_n,\boldsymbol a_n) - \int_{\sX} f(y) \, p(dy|x,\boldsymbol a_n) \right|.
\end{align*}
Therefore, it is sufficient to prove that $l_n(\boldsymbol a_n)$ converges to $0$. To achieve a contradiction, suppose that $l_n(\boldsymbol a_n)$ does not converge to $0$. Then, for some $\varepsilon >0$, there is a sub-sequence $l_{n_k}(\boldsymbol a_{n_k})$ of $l_n(\boldsymbol a_n)$ such that $|l_{n_k}(\boldsymbol a_{n_k})| > \varepsilon$ for all $k$. By compactness of $\boldsymbol \sA$, there is a further sub-sequence $l_{n_{k_l}}(\boldsymbol a_{n_{k_l}})$ such that $\boldsymbol a_{n_{k_l}}$ converges to some $\boldsymbol a$. By \cite[Theorem 2.4]{Ser82},  $l_{n_{k_l}}(\boldsymbol a_{n_{k_l}})$ converges to $0$. Contradiction.
\end{proof}

Now using this result let us prove some fundamental results related to the Bellman operators for the original model and the finite model. To this end, we define some operators in the sequel. These operators are all defined from $B(\sX)$ into $B(\sX)$. 

For each $\delta >0$, let $\boldsymbol \pi_{\delta}$ be a collection of stationary policies for the $\delta$-approximation finite game model. With an abuse of notation, let us also denote the extended policy $\boldsymbol \pi_{\delta} \circ Q_{\delta}$ via $\boldsymbol \pi_{\delta}$ as well. For each $i=1,\ldots,N$, define now the operators from $B(\sX)$ into $B(\sX)$: 
\small
\begin{align*}
	T^{\boldsymbol \pi_{\delta}^{-i}} J(z) &\coloneqq \min_{a^i \in \sA^i} \left[ c^i(z,\boldsymbol \pi_{\delta}^{-i}(z),a^i) + \int_{\sX} J(y) \, p(dy|z,\boldsymbol \pi_{\delta}^{-i}(z),a^i) \right] \\
\hat T^{\boldsymbol \pi_{\delta}^{-i}}_{\delta} J(z) 
	&\coloneqq \min_{a^i \in \sA^i_{\delta}} \biggl\{ \int_{\S_{\delta}^{i(z)}} \biggl[ c^i(x,\boldsymbol \pi^{-i}_{\delta}(z),a^i) \\
	&\phantom{xxxxxxxxxxxxxxxx} + \int_{\sX} \hat J(y) \, p(dy|x,\boldsymbol \pi^{-i}_{\delta}(z),a^i) \biggr] \, \nu_{\delta}^{i(z)}(dx) \biggr\}, 
\end{align*}
\normalsize
where $i: \sX \rightarrow \{1,\ldots,k_{\delta}\}$ gives the index of the bin to which $z$ belongs and $\hat J \coloneqq J \circ Q_{\delta}$. In the sequel, we use $"$$\,\,\widehat{ }\,\,\,$$"$ to denote piece-wise constant extensions of functions defined on $\sX_{\delta}$ to $\sX$. Let us prove our first result.

\begin{lemma} \label{result2}
Let $\{J_{\delta}\}_{\delta}$ be a family of functions defined on $\sX_{\delta}$ that is uniformly bounded in the sup-norm, and suppose that $\hat J_{\delta}$ converges pointwise to $J$ as $\delta \rightarrow 0$. Then, for every $z\in\sX$,
$$
\lim_{\delta \rightarrow 0} \left|\hat T^{\boldsymbol \pi_{\delta}^{-i}}_{\delta} \hat J_{\delta}(z) - T^{\boldsymbol \pi_{\delta}^{-i}}  J(z) \right| = 0.
$$
\end{lemma}

\begin{proof}
The result follows from Lemma~\ref{result1} using triangle inequality. Indeed, we have
\small
\begin{align*}
	&\left|\hat T^{\boldsymbol \pi^{-i}_{\delta}}_{\delta} \hat J_{\delta}(z)- T^{\boldsymbol \pi^{-i}_{\delta}} J(z) \right| \\
	&=\bigg| \min_{a^i \in \sA^i_{\delta}} \left\{ \int_{\S_{\delta}^{i(z)}} \left[ c^i(x,\boldsymbol \pi^{-i}_{\delta}(z),a^i) + \int_{\sX} \hat J_{\delta}(y) \, p(dy|x,\boldsymbol \pi^{-i}_{\delta}(z),a^i) \right] \, \nu_{\delta}^{i(z)}(dx) \right\} \\
	&\phantom{xxxxxxxxxxxxxxxxxxxxxxxxxxx} - \min_{a^i \in \sA^i} \left[ c^i(z,\boldsymbol \pi^{-i}_{\delta}(z),a^i) + \int_{\sX} J(y) \, p(dy|z,\boldsymbol \pi^{-i}_{\delta}(z),a^i) \right] \bigg| \\
	&\leq \bigg| \min_{a^i \in \sA^i_{\delta}} \left\{ \int_{\S_{\delta}^{i(z)}} \left[ c^i(x,\boldsymbol \pi^{-i}_{\delta}(z),a^i) + \int_{\sX} \hat J_{\delta}(y) \, p(dy|x,\boldsymbol \pi^{-i}_{\delta}(z),a^i) \right] \, \nu_{\delta}^{i(z)}(dx) \right\} \\
	&\phantom{xxxxxxxxx} - \min_{a^i \in \sA^i} \left\{ \int_{\S_{\delta}^{i(z)}} \left[ c^i(x,\boldsymbol \pi^{-i}_{\delta}(z),a^i) + \int_{\sX} \hat J_{\delta}(y) \, p(dy|x,\boldsymbol \pi^{-i}_{\delta}(z),a^i) \right] \, \nu_{\delta}^{i(z)}(dx) \right\} \bigg| \\
	&\phantom{x}+\bigg| \min_{a^i \in \sA^i} \left\{ \int_{\S_{\delta}^{i(z)}} \left[ c^i(x,\boldsymbol \pi^{-i}_{\delta}(z),a^i) + \int_{\sX} \hat J_{\delta}(y) \, p(dy|x,\boldsymbol \pi^{-i}_{\delta}(z),a^i) \right] \, \nu_{\delta}^{i(z)}(dx) \right\} \\
	&\phantom{xxxxxxxxx} - \min_{a^i \in \sA^i} \left[ c^i(z,\boldsymbol \pi^{-i}_{\delta}(z),a^i) + \int_{\sX} J(y) \, p(dy|z,\boldsymbol \pi^{-i}_{\delta}(z),a^i) \right] \bigg| \\
	&\leq \sup_{a^i \in \sA^i} \bigg|  \int_{\S_{\delta}^{i(z)}} \left[ c^i(x,\boldsymbol \pi^{-i}_{\delta}(z),Q_{i,\delta}(a^i)) + \int_{\sX} \hat J_{\delta}(y) \, p(dy|x,\boldsymbol \pi^{-i}_{\delta}(z),Q_{i,\delta}(a^i)) \right] \, \nu_{\delta}^{i(z)}(dx)  \\
	&\phantom{xxxxxxxxx} -  \int_{\S_{\delta}^{i(z)}} \left[ c^i(x,\boldsymbol \pi^{-i}_{\delta}(z),a^i) + \int_{\sX} \hat J_{\delta}(y) \, p(dy|x,\boldsymbol \pi^{-i}_{\delta}(z),a^i) \right] \, \nu_{\delta}^{i(z)}(dx) \bigg| \\
	&\phantom{x}+ \sup_{a^i \in \sA^i} \bigg| \int_{\S_{\delta}^{i(z)}} c^i(x,\boldsymbol \pi^{-i}_{\delta}(z),a^i) \, \nu_{\delta}^{i(z)}(dx) - c^i(z,\boldsymbol \pi^{-i}_{\delta}(z),a^i) \bigg| \\
	&\phantom{x}+ \sup_{a^i \in \sA^i} \bigg| \int_{\S_{\delta}^{i(z)}} \int_{\sX} \hat J_{\delta}(y) \, p(dy|x,\boldsymbol \pi^{-i}_{\delta}(z),a^i) \, \nu_{\delta}^{i(z)}(dx) \\
	&\phantom{xxxxxxxxxxxxxxxxxxxxxxxxxxxxxxxxxx}-  \int_{\sX} J(y) \, p(dy|z,\boldsymbol \pi^{-i}_{\delta}(z),a^i)  \bigg|. 
\end{align*}
\normalsize

\noindent First, note that in the last expression, the first term was originally minimized over $\sA^i_{\delta}$, so at first sight, taking $\sup_{a^i\in\sA^i}$ may seem incorrect. However, minimizing over $\sA^i_{\delta}$ is equivalent to minimizing over $\sA^i$ after replacing each $a^i$ by its quantized version $Q_{i,\delta}(a^i)$. Thus, we can write the minimization over $\sA^i$ and use $Q_{i,\delta}(a^i)$ in place of $a^i$.

The second term in the last expression converges to zero as $\delta\to0$ by the uniform continuity of $c^i(x,\boldsymbol a)$. The last term converges to zero by Lemma~\ref{result1}. Finally, the first term can be upper bounded by
\small
\begin{align*}
	&\sup_{z \in \sX} \sup_{a^i \in \sA^i} \bigg|  c^i(z,\boldsymbol \pi^{-i}_{\delta}(z),Q_{i,\delta}(a^i)) + \int_{\sX} \hat J_{\delta}(y) \, p(dy|z,\boldsymbol \pi^{-i}_{\delta}(z),Q_{i,\delta}(a^i))   \\
&\phantom{xxxxxxxxxxxxxxxxxxxxxxxxx}	-  c^i(z,\boldsymbol \pi^{-i}_{\delta}(z),a^i) - \int_{\sX} \hat J_{\delta}(y) \, p(dy|z,\boldsymbol \pi^{-i}_{\delta}(z),a^i)  \bigg| \\ 
	&\leq \sup_{z \in \sX} \sup_{a^i \in \sA^i} \bigg|  c^i(z,\boldsymbol \pi^{-i}_{\delta}(z),Q_{i,\delta}(a^i)) - c^i(z,\boldsymbol \pi^{-i}_{\delta}(z),a^i) \bigg| \\
	&+ \sup_{z \in \sX} \sup_{a^i \in \sA^i} \bigg| \int_{\sX} \hat J_{\delta}(y) \, p(dy|z,\boldsymbol \pi^{-i}_{\delta}(z),Q_{i,\delta}(a^i)) \\
	&\phantom{xxxxxxxxxxxxxxxxxxxxxxxxx}- \int_{\sX} \hat J_{\delta}(y) \, p(dy|z,\boldsymbol \pi^{-i}_{\delta}(z),a^i)  \bigg|.
\end{align*}
\normalsize
The first term above converges to zero as $\delta \rightarrow 0$ again by uniform continuity of the function $c^i(x,\boldsymbol a)$ and the second term converges to zero by Lemma~\ref{result1}. This completes the proof. 
\end{proof}

We now prove similar result about Bellman operators that give the value functions of the joint policy $\boldsymbol \pi_{\delta}$ both for the finite game model and the original game model. We first define the following operator on $B(\sX)$:
$$
T^{\boldsymbol \pi_{\delta},i} J(z) \coloneqq c^i(z,\boldsymbol \pi_{\delta}(z)) + \beta \, \int_{\sX} J(y) \, p(dy|z,\boldsymbol \pi_{\delta}(z)). 
$$
It is trivial to prove that $T^{\boldsymbol \pi_{\delta},i}$ is $\beta$-contraction and the unique fixed point of it is the value function $J^i(\boldsymbol \pi_{\delta}, \cdot)$ of agent~$i$ under the joint policy $\boldsymbol \pi_{\delta}$ in the original game model. We also define the following operator on $B(\sX)$:
$$
\hat T^{\boldsymbol \pi_{\delta},i}_{\delta} J(z) \coloneqq \int_{\S_{\delta}^{i(z)}} \left[ c^i(x,\boldsymbol \pi_{\delta}(z)) + \beta \, \int_{\sX} \hat J(y) \, p(dy|x,\boldsymbol \pi_{\delta}(z)) \right] \, \nu_{\delta}^{i(z)}(dx),
$$
where recall that $\hat J := J \circ Q_{\delta}$. This operator is also $\beta$-contraction and the unique fixed point of it is the value function $\hat J_{\delta}^i(\boldsymbol \pi_{\delta}; \cdot)$ of agent~$i$ under the joint policy $\boldsymbol \pi_{\delta}$ in the finite game model.

\begin{lemma} \label{result3}
Let $\{J_{\delta}\}_{\delta}$ be a family of functions defined on $\sX_{\delta}$ that is uniformly bounded in the sup-norm, and suppose that $\hat J_{\delta}$ converges pointwise to $J$ as $\delta \rightarrow 0$. Then, for every $z\in\sX$,
$$
\lim_{\delta \rightarrow 0} \left|\hat T^{\boldsymbol \pi_{\delta},i}_{\delta} \hat J_{\delta}(z) - T^{\boldsymbol \pi_{\delta},i}  J(z) \right| = 0.
$$
\end{lemma}

\begin{proof}
The proof is very similar to the proof of Lemma~\ref{result2}  and follows from Lemma~\ref{result1} using triangle inequality. Indeed, we have
\small
\begin{align*}
	&\left|\hat T^{\boldsymbol \pi_{\delta},i}_{\delta} \hat J_{\delta}(z)- T^{\boldsymbol \pi_{\delta},i} J(z) \right| \\
	&=\bigg| \left\{ \int_{\S_{\delta}^{i(z)}} \left[ c^i(x,\boldsymbol \pi_{\delta}(z)) + \int_{\sX} \hat J_{\delta}(y) \, p(dy|x,\boldsymbol \pi_{\delta}(z)) \right] \, \nu_{\delta}^{i(z)}(dx) \right\} \\
	&\phantom{xxxxxxxxxxxxxxxxxxxxxxxxxxx} - \left[ c^i(z,\boldsymbol \pi_{\delta}(z)) + \int_{\sX} J(y) \, p(dy|z,\boldsymbol \pi_{\delta}(z)) \right] \bigg| \\
	&\leq \bigg| \int_{\S_{\delta}^{i(z)}} c^i(x,\boldsymbol \pi_{\delta}(z)) \, \nu_{\delta}^{i(z)}(dx) - c^i(z,\boldsymbol \pi_{\delta}(z)) \bigg| \\
	&\phantom{x}+ \bigg| \int_{\S_{\delta}^{i(z)}} \int_{\sX} \hat J_{\delta}(y) \, p(dy|x,\boldsymbol \pi_{\delta}(z)) \, \nu_{\delta}^{i(z)}(dx) \\
	&\phantom{xxxxxxxxxxxxxxxxxxxxxxxxxxxxxxxxxx}-  \int_{\sX} J(y) \, p(dy|z,\boldsymbol \pi_{\delta}(z))  \bigg|. 
\end{align*}
\normalsize

\noindent In the last expression, the first term converges to zero as $\delta \rightarrow 0$ by uniform continuity of the function $c^i(x,\boldsymbol a)$. The second term converges to zero by Lemma~\ref{result1}. This completes the proof. 
\end{proof}

\section{Finite-Horizon Cost}\label{finite_horizon}

Here, we consider the approximation problem for the finite-horizon cost criterion. Throughout this section, Assumption~\ref{compact:as1} is assumed to hold.

Given some fixed Markov policies $\boldsymbol \pi^{-i}$ of all players except player~$i$, the player~$i$ best response is characterized via dynamic programming principle: 
\begin{align}\label{optimality_eq}
	T^{\boldsymbol \pi^{-i}}_t J_{t+1}^*(\boldsymbol \pi^{-i}; \cdot) = J_{t}^*(\boldsymbol \pi^{-i}; \cdot), \,\, \text{for} \,\, t=0,\ldots,T-1, 
\end{align}
where the operator $T^{\boldsymbol \pi^{-i}}_t: B(\sX) \rightarrow B(\sX)$ is defined as 
\begin{align*}
	T^{\boldsymbol \pi^{-i}}_t J(x) &\coloneqq \min_{a^i \in \sA^i} \left[ c^i(x,\boldsymbol \pi^{-i}_t(x),a^i) + \int_{\sX} J(y) \, p(dy|x,\boldsymbol \pi^{-i}_t(x),a^i) \right] \\
	&=\min_{\gamma^i \in \P(\sA^i)} \left[ c^i(x,\boldsymbol \pi^{-i}_t(x),\gamma^i) + \int_{\sX} J(y) \, p(dy|x,\boldsymbol \pi^{-i}_t(x),\gamma^i) \right]
\end{align*}
and $J_{T}^*(\boldsymbol \pi^{-i}; \cdot) = 0$. Here, with an abuse of notation, for any collection of probability measures $(\gamma^1,\ldots,\gamma^N) \in \P(\sA^1) \times \ldots \times \P(\sA^N)$, we define 
\begin{align*}
	c^i(x,\gamma^1,\ldots,\gamma^N) &\coloneqq \int_{\boldsymbol \sA} c^i(x,a^1,\ldots,a^N) \, \gamma^1(da^1) \otimes \ldots \otimes \gamma^N(da^N) \\ \\
	p(\cdot|x,\gamma^1,\ldots,\gamma^N) &\coloneqq \int_{\boldsymbol \sA} p(\cdot|x,a^1,\ldots,a^N) \, \gamma^1(da^1) \otimes \ldots \otimes \gamma^N(da^N).
\end{align*}
In above recursion, $J_{t}^*(\boldsymbol \pi^{-i}; \cdot)$ is the optimal cost-to-go at time time $t$ of the player~$i$, if the policies of other players are fixed as $\boldsymbol \pi^{-i}$: 
\begin{align}\label{optimal_cost}
	J_{t}^*(\boldsymbol \pi^{-i}; x) \coloneqq \inf_{\pi^i \in \Pi^i} \cE^{(\boldsymbol \pi^{-i},\pi^i)}_x \left[ \sum_{l=t}^{T-1} c^i(x_l,\boldsymbol a_l) \right], 
\end{align}
where subscript $x$ in the expectation means that $x_t =x$ and the stochastic kernels after time $t$ in the policies are used only. For instance, $\{\pi^i_l\}_{l=0}^{t-1}$ are irrelevant in this case. If the measurable functions $\pi_t^{*,i}(\cdot;\boldsymbol \pi^{-i})$ ($t=0,\ldots,T-1$) from $\sX$ to $\P(\sA^i)$ minimize the left part in (\ref{optimality_eq}) for all $x \in \sX$, then it is known that the Markov policy $\pi^{*,i}(\cdot;\boldsymbol \pi^{-i}) = \{\pi^{*,i}_t(\cdot;\boldsymbol \pi^{-i})\}_{t=0}^{T-1}$ is the optimal solution of the optimization problem in (\ref{optimal_cost}) for each $t=0,\ldots,T-1$ (again for each $t$, functions before time $t$ are irrelevant). 

For the approximate finite model, similar definitions can be made if we replace 
$\bigl( \sX, \sA^1,\ldots,$ $ \sA^N, c^1,\ldots, c^N, p \bigr)$ with $\bigl( \sX_{\delta}, \sA^1_{\delta},\ldots, \sA^N_{\delta}, c^1_{\delta},\ldots, c^N_{\delta}, p_{\delta} \bigr)$ and integral with summation. In this case, we also add $\delta$ as a subscript to the operators and the optimal cost-to-go functions. 

\subsubsection*{Existence of Approximate Markov Nash Equilibrium}

For any $\delta > 0$, by \cite{rieder1979equilibrium}, it is known that there exists a Markov Nash equilibrium $\boldsymbol \pi^*_{\delta}$ for the finite $\delta$-approximation of the original game problem. Hence, for all $i=1,\ldots,N$, we have 
\begin{align}\label{optimality_eq_finite}
	T^{\boldsymbol \pi^{*,-i}_{\delta}}_{\delta,t} J_{\delta,t+1}^*(\boldsymbol \pi^{*,-i}_{\delta}; \cdot) = J_{\delta,t}^*(\boldsymbol \pi^{*,-i}_{\delta}; \cdot), \,\, \text{for} \,\, t=0,\ldots,T-1, 
\end{align}
where the operator $T^{\boldsymbol \pi^{*,-i}_{\delta}}_{\delta,t}: B(\sX_{\delta}) \rightarrow B(\sX_{\delta})$ is defined as 
\small
\begin{align*}
	T^{\boldsymbol \pi^{*,-i}_{\delta}}_{\delta,t} J(x_j)
	&\coloneqq \min_{a^i \in \sA^i_{\delta}} \biggl\{ \int_{\S_{\delta}^j} \biggl[ c^i(x,\boldsymbol \pi^{*,-i}_{\delta,t}(x_j),a^i) \\
	&\phantom{xxxxxxxxxxxxxxxx} + \int_{\sX} \hat J(y) \, p(dy|x,\boldsymbol \pi^{*,-i}_{\delta,t}(x_j),a^i) \biggr] \, \nu_{\delta}^j(dx) \biggr\} \\
	&= \min_{\gamma^i \in \P(\sA^i_{\delta})} \biggl\{ \int_{\S_{\delta}^j} \biggl[ c^i(x,\boldsymbol \pi^{*,-i}_{\delta,t}(x_j),\gamma^i) \\
	&\phantom{xxxxxxxxxxxxxxxx} + \int_{\sX} \hat J(y) \, p(dy|x,\boldsymbol \pi^{*,-i}_{\delta,t}(x_j),\gamma^i) \biggr] \, \nu_{\delta}^j(dx) \biggr\},
\end{align*}
\normalsize
$\hat J = J \circ Q_{\delta}$, and $J_{\delta,T}^*(\boldsymbol \pi^{*,-i}_{\delta}; \cdot) = 0$. For each $i=1,\ldots,N$, the minimum in the left part of (\ref{optimality_eq_finite}) is achieved by the policies in Markov perfect Nash equilibrium $\boldsymbol \pi^*_{\delta}$. 

We now use the operators $\hat T^{\boldsymbol \pi^{*,-i}_{\delta}}_{\delta,t}$ on $B(\sX)$, that are defined in Section~\ref{fundamental}, as an extension to $T^{\boldsymbol \pi^{*,-i}_{\delta}}_{\delta,t}$:
\small
\begin{align*}
	\hat T^{\boldsymbol \pi^{*,-i}_{\delta}}_{\delta,t} J(z)
	&\coloneqq \min_{a^i \in \sA^i_{\delta}} \biggl\{ \int_{\S_{\delta}^{i(z)}} \biggl[ c^i(x,\boldsymbol \pi^{*,-i}_{\delta,t}(z),a^i) \\
	&\phantom{xxxxxxxxxxxxxxxx} + \int_{\sX} \hat J(y) \, p(dy|x,\boldsymbol \pi^{*,-i}_{\delta,t}(z),a^i) \biggr] \, \nu_{\delta}^{i(z)}(dx) \biggr\} \\
	&= \min_{\gamma^i \in \P(\sA^i_{\delta})} \biggl\{ \int_{\S_{\delta}^{i(z)}} \biggl[ c^i(x,\boldsymbol \pi^{*,-i}_{\delta,t}(z),\gamma^i) \\
	&\phantom{xxxxxxxxxxxxxxxx} + \int_{\sX} \hat J(y) \, p(dy|x,\boldsymbol \pi^{*,-i}_{\delta,t}(z),\gamma^i) \biggr] \, \nu_{\delta}^{i(z)}(dx) \biggr\},
\end{align*}
\normalsize
$\hat J = J \circ Q_{\delta}$, and recall that with an abuse of notation, we denote the extended policy $\boldsymbol \pi^{*,-i}_{\delta,t} \circ Q_{\delta}$ as $\boldsymbol \pi^{*,-i}_{\delta,t}$ in order not to complicate the notation further. Recall also that $i: \sX \rightarrow \{1,\ldots,k_{\delta}\}$ gives the index of the bin to which $z$ belongs. One can prove that  
\begin{align}\label{optimality_eq_finite_extended}
	\hat T^{\boldsymbol \pi^{*,-i}_{\delta}}_{\delta,t} \hat J_{\delta,t+1}^*(\boldsymbol \pi^{*,-i}_{\delta}; \cdot) = \hat J_{\delta,t}^*(\boldsymbol \pi^{*,-i}_{\delta}; \cdot), \,\, \text{for} \,\, t=0,\ldots,T-1, 
\end{align}
where recall that $"$$\,\,\widehat{ }\,\,\,$$"$ means piece-wise constant extensions of functions defined on $\sX_{\delta}$ to $\sX$. 

We now present the main results of this section, which follow from those established in Section~\ref{fundamental}.

\begin{theorem}\label{finite_horizon_thm_1}
	Under Assumption~\ref{compact:as1}, for each $i=1,\ldots,N$, we have 
	$$
	\lim_{\delta \rightarrow 0} |\hat J_{\delta,t}^*(\boldsymbol \pi^{*,-i}_{\delta}; z)-J_{t}^*(\boldsymbol \pi^{*,-i}_{\delta}; z)| = 0, \,\, \forall z \in \sX \,\,, \forall  \, t=0,\ldots,T.
	$$
\end{theorem}

\begin{proof}

We prove the result by backward induction. Since the terminal cost at time $T$ is zero for each problem, the result trivially holds. 

Suppose that the statement is true for $t+1$ and consider $t$. Then, for $z \in \sX$, we have
\small
\begin{align*}
	&|\hat J_{\delta,t}^*(\boldsymbol \pi^{*,-i}_{\delta}; z)-J_{t}^*(\boldsymbol \pi^{*,-i}_{\delta}; z)| = |\hat T^{\boldsymbol \pi^{*,-i}_{\delta}}_{\delta,t} \hat J_{\delta,t+1}^*(\boldsymbol \pi^{*,-i}_{\delta}; z)- T^{\boldsymbol \pi^{*,-i}_{\delta}}_{t} J_{t+1}^*(\boldsymbol \pi^{*,-i}_{\delta}; z)|. 
\end{align*}
\normalsize
The last expression converges to zero by Lemma~\ref{result2} and the fact that $\hat J_{\delta,t+1}^*(\boldsymbol \pi^{*,-i}_{\delta}; z)$ converges to $J_{t+1}^*(\boldsymbol \pi^{*,-i}_{\delta}; z)$ for all $z \in \sX$ by the induction hypothesis. This completes the proof. 
\end{proof}

We now proceed to prove the key result of this section, which implies that the Markov Nash equilibrium of the finite model, when extended to the original model, serves as an approximate Markov Nash equilibrium for the original game.

\begin{theorem}\label{finite_horizon_thm_2}
	Under Assumption~\ref{compact:as1}, for each $i=1,\ldots,N$, we have 
	$$
	\lim_{\delta \rightarrow 0} |J_{t}^i(\boldsymbol \pi^{*}_{\delta}, z)-J_{t}^*(\boldsymbol \pi^{*,-i}_{\delta}; z)| = 0, \,\, \forall z \in \sX, \,\, \forall  \, t=0,\ldots,T, 
	$$
	where $J_{t}^i(\boldsymbol \pi^{*}_{\delta}, \cdot)$ is the cost-to-go of player~$i$ at time $t$ under the joint policy $\boldsymbol \pi^{*}_{\delta}$. 
\end{theorem}

\begin{proof}

We again prove the result by backward induction. Since the terminal cost at time $T$ is zero, the result trivially holds. 

Suppose that the statement is true for $t+1$ and consider $t$. Then, 
\begin{align}\label{thm_2_eq_1}
	&|J_{t}^i(\boldsymbol \pi^{*}_{\delta}, z)-J_{t}^*(\boldsymbol \pi^{*,-i}_{\delta}; z)| = | T^{\boldsymbol \pi^{*}_{\delta},i}_{t}  J_{t+1}^i(\boldsymbol \pi^{*}_{\delta}, z)- T^{\boldsymbol \pi^{*,-i}_{\delta}}_{t} J_{t+1}^*(\boldsymbol \pi^{*,-i}_{\delta}; z)|, 
\end{align}
where 
$$
T^{\boldsymbol \pi^{*}_{\delta},i}_{t} J(z) \coloneqq c^i(z,\boldsymbol \pi^{*}_{\delta,t}(z)) + \int_{\sX} J(y) \, p(dy|z,\boldsymbol \pi^{*}_{\delta,t}(z)). 
$$
Then, we can bound (\ref{thm_2_eq_1}) via triangle inequality as follows
\begin{align}		
	(\ref{thm_2_eq_1}) \leq & | T^{\boldsymbol \pi^{*}_{\delta},i}_{t}  J_{t+1}^i(\boldsymbol \pi^{*}_{\delta}, z)- \hat T^{\boldsymbol \pi^{*,-i}_{\delta}}_{\delta,t} \hat J_{\delta,t+1}^*(\boldsymbol \pi^{*,-i}_{\delta}; z)| \nonumber \\
	&+ |\hat T^{\boldsymbol \pi^{*,-i}_{\delta}}_{\delta,t} \hat J_{\delta,t+1}^*(\boldsymbol \pi^{*,-i}_{\delta}; z)-T^{\boldsymbol \pi^{*,-i}_{\delta}}_{t} J_{t+1}^*(\boldsymbol \pi^{*,-i}_{\delta}; z)|
\label{thm_2_eq_2}
\end{align}
The second term in the last expression converges to zero as $\delta \rightarrow 0$ by Theorem~\ref{finite_horizon_thm_1} and Lemma~\ref{result2}. The first term in the last expression converges to zero as $\delta \rightarrow 0$ by Theorem~\ref{finite_horizon_thm_1} and Lemma~\ref{result3}. This completes the proof.  
\end{proof}

Note that for each $i=1,\ldots,N$, we have
\begin{align*}
	J_{0}^*(\boldsymbol \pi^{*,-i}_{\delta}; x) = \inf_{\pi \in \Pi^i} \cE^{(\boldsymbol \pi^{*,-i}_{\delta},\pi^i)}_x \left[ \sum_{t=0}^{T-1} c^i(x_t,\boldsymbol a_t) \right].  
\end{align*}
Hence, Theorem~\ref{finite_horizon_thm_2} implies that for any $\varepsilon > 0$ and $x \in \sX$, there exists $\delta(\varepsilon,x)$ such that for any $\delta \leq \delta(\varepsilon,x)$, the policy $\boldsymbol \pi^{*}_{\delta}$ is Markov Nash $\varepsilon$-equilibrium for the initial point $x \in \sX$. This establishes the existence of an approximate Markov Nash equilibrium for the finite-horizon Borel model.

\section{Discounted Cost}\label{discounted}

In this section, we address the approximation problem related to the discounted cost criterion. We assume that Assumption~\ref{compact:as1} remains valid throughout this section. 

Given some fixed stationary policies $\boldsymbol \pi^{-i}$ of all players except player~$i$, the player~$i$ best response is characterized via dynamic programming principle: 
\begin{align}\label{optimality_eq_discounted}
	T^{\boldsymbol \pi^{-i}} J^*(\boldsymbol \pi^{-i}; \cdot) = J^*(\boldsymbol \pi^{-i}; \cdot), 
\end{align}
where the operator $T^{\boldsymbol \pi^{-i}}: B(\sX) \rightarrow B(\sX)$ is defined as 
$$
T^{\boldsymbol \pi^{-i}} J(x) \coloneqq \min_{a^i \in \sA^i} \left[ c^i(x,\boldsymbol \pi^{-i}(x),a^i) + \beta \, \int_{\sX} J(y) \, p(dy|x,\boldsymbol \pi^{-i}(x),a^i) \right].
$$
Here, similar to the operators defined for the finite-horizon cost criterion, one can always perform the minimization over the set of probability measures on the action spaces, i.e. $\P(\sA^i)$, which yields the same operator. However, to avoid complicating the notation, we omit explicitly stating this each time we define such an operator, but we implicitly assume its ability to be defined and minimized over the set of probability measures as well.

It is straightforward to prove that $T^{\boldsymbol \pi^{-i}}$ is $\beta$-contraction on $B(\sX)$ with respect to sup-norm. Hence, it has a unique fixed point by Banach fixed point theorem and this unique fixed point $J^*(\boldsymbol \pi^{-i}; \cdot)$ is the optimal cost of the player~$i$, if the policies of other players are fixed as $\boldsymbol \pi^{-i}$: 
\begin{align}\label{optimal_cost_discounted}
	J^*(\boldsymbol \pi^{-i}; x) \coloneqq \inf_{\pi^i \in \Pi^i} \cE^{(\boldsymbol \pi^{-i},\pi^i)}_x \left[ \sum_{t=0}^{\infty} \beta^t c^i(x_t,\boldsymbol a_t) \right].
\end{align}
If the measurable function $\pi^{*,i}(\cdot;\boldsymbol \pi^{-i})$ from $\sX$ to $\P(\sA^i)$ minimizes the expression in (\ref{optimality_eq_discounted}) for all $x \in \sX$, then it is known that the stationary policy $\pi^{*,i}(\cdot;\boldsymbol \pi^{-i})$ is the optimal solution of the optimization problem in (\ref{optimal_cost_discounted}). 

For the approximate finite model, similar definitions can be made if we replace 
$\bigl( \sX, \sA^1,\ldots,$  $\sA^N, c^1,\ldots, c^N, p \bigr)$ with $\bigl( \sX_{\delta}, \sA^1_{\delta},\ldots, \sA^N_{\delta}, c^1_{\delta},\ldots, c^N_{\delta}, p_{\delta} \bigr)$ and integral with summation. In this case, we also add $\delta$ as a subscript to the operators and the optimal cost function. 

\subsubsection*{Existence of Approximate Stationary Nash Equilibrium}

For any $\delta > 0$, by \cite{Fin64}, it is known that there exists a stationary Nash equilibrium $\boldsymbol \pi^*_{\delta}$ for the finite $\delta$-approximation of the original game problem. Hence, for all $i=1,\ldots,N$, we have 
\begin{align}\label{optimality_eq_finite_discounted}
	T^{\boldsymbol \pi^{*,-i}_{\delta}}_{\delta} J_{\delta}^*(\boldsymbol \pi^{*,-i}_{\delta}; \cdot) = J_{\delta}^*(\boldsymbol \pi^{*,-i}_{\delta}; \cdot), 
\end{align}
where the operator $T^{\boldsymbol \pi^{*,-i}_{\delta}}_{\delta}: B(\sX_{\delta}) \rightarrow B(\sX_{\delta})$ is defined as 

\small
$$
\begin{aligned}
T^{\boldsymbol \pi^{*,-i}_{\delta}}_{\delta} J(x_j) &\coloneqq \min_{a^i \in \sA^i_{\delta}} \biggl\{ \int_{\S_{\delta}^j} \biggl[ c^i(x,\boldsymbol \pi^{*,-i}_{\delta}(x_j),a^i) \\
&\phantom{xxxxxxxxxxxxxxxx} + \beta \, \int_{\sX} \hat J(y) \, p(dy|x,\boldsymbol \pi^{*,-i}_{\delta}(x_j),a^i) \biggr] \biggr\}, 
\end{aligned}
$$
\normalsize
where $\hat J = J \circ Q_{\delta}$. For each $i=1,\ldots,N$, the minimum in (\ref{optimality_eq_finite_discounted}) is achieved by the policies in stationary Nash equilibrium $\boldsymbol \pi^*_{\delta}$. 

The extended version of the operator $T^{\boldsymbol \pi^{*,-i}_{\delta}}_{\delta}$ to $B(\sX)$ is as follows:
\small
\begin{align*}
	\hat T^{\boldsymbol \pi^{*,-i}_{\delta}}_{\delta} J(z)  &\coloneqq \min_{a^i \in \sA^i_{\delta}} \biggl\{ \int_{\S_{\delta}^{i(z)}} \biggl[ c^i(x,\boldsymbol \pi^{*,-i}_{\delta}(z),a^i) \\
	&\phantom{xxxxxxxxxxxxxxxx} + \beta \, \int_{\sX} \hat J(y) \, p(dy|x,\boldsymbol \pi^{*,-i}_{\delta}(z),a^i) \biggr] \, \nu_{\delta}^{i(z)}(dx) \biggr\}, 
\end{align*}
\normalsize
where $\hat J = J \circ Q_{\delta}$, and with again abuse of notation, we denote the extended policy $\boldsymbol \pi^{*,-i}_{\delta} \circ Q_{\delta}$ as $\boldsymbol \pi^{*,-i}_{\delta}$ in order not to complicate the notation further. Recall that $i: \sX \rightarrow \{1,\ldots,k_{\delta}\}$ gives the index of the bin to which $z$ belongs. One can prove that  
\begin{align}\label{optimality_eq_finite_extended_discounted}
	\hat T^{\boldsymbol \pi^{*,-i}_{\delta}}_{\delta} \hat J_{\delta}^*(\boldsymbol \pi^{*,-i}_{\delta}; \cdot) = \hat J_{\delta}^*(\boldsymbol \pi^{*,-i}_{\delta}; \cdot), 
\end{align}
where we recall that $"$$\,\,\widehat{ }\,\,\,$$"$ means piece-wise constant extensions of functions defined on $\sX_{\delta}$ to $\sX$. 

We now present the main results of this section, which follow from those established in Section~\ref{fundamental}.

\begin{theorem}\label{discounted_thm_1}
	Under Assumption~\ref{compact:as1}, for each $i=1,\ldots,N$, we have 
	$$
	\lim_{\delta \rightarrow 0} |\hat J_{\delta}^*(\boldsymbol \pi^{*,-i}_{\delta}; z)-J^*(\boldsymbol \pi^{*,-i}_{\delta}; z)| = 0, \,\, \forall z \in \sX.
	$$
\end{theorem}

\begin{proof}

To prove the result, we use contraction property of the operators $\hat T^{\boldsymbol \pi^{*,-i}_{\delta}}_{\delta}$ and $T^{\boldsymbol \pi^{*,-i}_{\delta}}$. Indeed, by Banach fixed point theorem, if we start with a common initial function $J_0 \in B(\sX)$, then 

\begin{align*}
	\left(\hat T^{\boldsymbol \pi^{*,-i}_{\delta}}_{\delta}\right)^t J_0 &\eqqcolon  \hat J_{\delta,t}^*(\boldsymbol \pi^{*,-i}_{\delta}; \cdot) \rightarrow \hat J_{\delta}^*(\boldsymbol \pi^{*,-i}_{\delta}; \cdot), \\
	\left( T^{\boldsymbol \pi^{*,-i}_{\delta}}\right)^t J_0 &\eqqcolon  J_{t}^*(\boldsymbol \pi^{*,-i}_{\delta}; \cdot) \rightarrow  J^*(\boldsymbol \pi^{*,-i}_{\delta}; \cdot)
\end{align*}
in sup-norm (hence pointwise as well) as $t \rightarrow \infty$. Hence, by using induction, we first prove that 
$$
\lim_{\delta \rightarrow 0} |\hat J_{\delta,t}^*(\boldsymbol \pi^{*,-i}_{\delta}; z)-J_t^*(\boldsymbol \pi^{*,-i}_{\delta}; z)| = 0, \,\, \forall z \in \sX,
$$
for all $t\geq0$. Then, the result follows from the triangle inequality.

Since the function $J_0$ at time zero is common, the result trivially holds. Suppose that the statement is true for $t$ and consider $t+1$. Then, for $z \in \sX$, we have
\begin{align*}
	&|\hat J_{\delta,t+1}^*(\boldsymbol \pi^{*,-i}_{\delta}; z)-J_{t+1}^*(\boldsymbol \pi^{*,-i}_{\delta}; z)| = |\hat T^{\boldsymbol \pi^{*,-i}_{\delta}}_{\delta} \hat J_{\delta,t}^*(\boldsymbol \pi^{*,-i}_{\delta}; z)- T^{\boldsymbol \pi^{*,-i}_{\delta}} J_{t}^*(\boldsymbol \pi^{*,-i}_{\delta}; z)|.
\end{align*}
The last expression converges to zero by Lemma~\ref{result2} and the fact that $\hat J_{\delta,t}^*(\boldsymbol \pi^{*,-i}_{\delta}; z)$ converges to $J_{t}^*(\boldsymbol \pi^{*,-i}_{\delta}; z)$ for all $z \in \sX$ by the induction hypothesis. This completes the proof. 
\end{proof}

We will now demonstrate the main result of this section, which implies that the stationary Nash equilibrium from the finite model, when applied to the original model, functions as an approximate stationary Nash equilibrium for the original game.

\begin{theorem}\label{discounted_thm_2}
	Under Assumption~\ref{compact:as1}, for each $i=1,\ldots,N$, we have 
	$$
	\lim_{\delta \rightarrow 0} |J^i(\boldsymbol \pi^{*}_{\delta}, z)-J^*(\boldsymbol \pi^{*,-i}_{\delta}; z)| = 0, \,\, \forall z \in \sX,
	$$
where $J^i(\boldsymbol \pi^{*}_{\delta}, \cdot)$ is the cost of player~$i$ under the joint policy $\boldsymbol \pi^{*}_{\delta}$. 
\end{theorem}

\begin{proof}

We prove the result by contraction property of the operators. We first define the following operator on $B(\sX)$:
$$
T^{\boldsymbol \pi^{*}_{\delta},i} J(z) \coloneqq c^i(z,\boldsymbol \pi^{*}_{\delta}(z)) + \beta \, \int_{\sX} J(y) \, p(dy|z,\boldsymbol \pi^{*}_{\delta}(z)). 
$$
It is trivial to prove that $T^{\boldsymbol \pi^{*}_{\delta},i}$ is $\beta$-contraction and the unique fixed point of it is $J^i(\boldsymbol \pi^{*}_{\delta}, \cdot)$. We also define the following operator on $B(\sX)$:
$$
\hat T^{\boldsymbol \pi^{*}_{\delta},i} J(z) \coloneqq \int_{\S_{\delta}^{i(z)}} \left[ c^i(x,\boldsymbol \pi^{*}_{\delta}(z)) + \beta \, \int_{\sX} \hat J(y) \, p(dy|x,\boldsymbol \pi^{*}_{\delta}(z)) \right] \, \nu_{\delta}^{i(z)}(dx). 
$$
This operator is also $\beta$-contraction and the unique fixed point of it is $\hat J_{\delta}^*(\boldsymbol \pi^{*,-i}_{\delta}; \cdot)$ since the minimum is achieved in $\hat T^{\boldsymbol \pi^{*,-i}_{\delta}}_{\delta} \hat J_{\delta}^*(\boldsymbol \pi^{*,-i}_{\delta}; \cdot)$ by $\pi_{\delta}^{*,i}$ as $\boldsymbol \pi^{*}_{\delta}$ is a Nash equilibrium in the finite game; that is
$$
\hat J_{\delta}^*(\boldsymbol \pi^{*,-i}_{\delta}; \cdot)  = \hat T^{\boldsymbol \pi^{*,-i}_{\delta}}_{\delta} \hat J_{\delta}^*(\boldsymbol \pi^{*,-i}_{\delta}; \cdot) = \hat T^{\boldsymbol \pi^{*}_{\delta},i} \hat J_{\delta}^*(\boldsymbol \pi^{*,-i}_{\delta}; \cdot).
$$
With these observations, for any $z \in \sX$, we then have 
\begin{align*}
	 &|J^i(\boldsymbol \pi^{*}_{\delta}, z)-J^*(\boldsymbol \pi^{*,-i}_{\delta}; z)\|  \\
	 &  \leq |T^{\boldsymbol \pi^{*}_{\delta},i} J^i(\boldsymbol \pi^{*}_{\delta}, z) - \hat T^{\boldsymbol \pi^{*}_{\delta},i} \hat J_{\delta}^*(\boldsymbol \pi^{*,-i}_{\delta}; z) | + | \hat J_{\delta}^*(\boldsymbol \pi^{*,-i}_{\delta}; z) - J^*(\boldsymbol \pi^{*,-i}_{\delta}; z)|.
\end{align*}
The second term in the last expression converges to zero as $\delta \rightarrow 0$ by Theorem~\ref{discounted_thm_1}. For the first term, note that $J^i(\boldsymbol \pi^{*}_{\delta}, \cdot)$ is the unique fixed point of the $\beta$-contraction operator $T^{\boldsymbol \pi^{*}_{\delta},i}$ and $\hat J_{\delta}^*(\boldsymbol \pi^{*,-i}_{\delta}; \cdot)$ is the unique fixed point of the $\beta$-contraction operator $\hat T^{\boldsymbol \pi^{*}_{\delta},i}$. Hence, using exactly the same arguments that we used in the proof of Theorem~\ref{discounted_thm_1} but in this case using Lemma~\ref{result3} in place of Lemma~\ref{result2}, we can establish that this term converges to zero as $\delta \rightarrow 0$. This completes the proof.  
\end{proof}

We next discuss the implications of Theorem~\ref{discounted_thm_2}. Note that for each $i=1,\ldots,N$, 
\begin{align*}
	J^*(\boldsymbol \pi^{*,-i}_{\delta}; x) = \inf_{\pi \in \Pi^i} \cE^{(\boldsymbol \pi^{*,-i}_{\delta},\pi^i)}_x \left[ \sum_{t=0}^{\infty} \beta^t \, c^i(x_t,\boldsymbol a_t) \right].  
\end{align*}
Hence, Theorem~\ref{discounted_thm_2} implies that for any $\varepsilon > 0$ and $x \in \sX$, there exists $\delta(\varepsilon,x)$ such that for any $\delta \leq \delta(\varepsilon,x)$, the policy $\boldsymbol \pi^{*}_{\delta}$ is stationary Nash $\varepsilon$-equilibrium for the initial point $x$. This establishes the existence of an approximate stationary Nash equilibrium for the finite-horizon Borel model.

\section{Extension to Non-Compact State Spaces}\label{noncompact}

In this section, we explain how the results established in the previous sections can be extended to non-compact state stochastic games. We employ the following strategy: (i) first, we define a sequence of compact-state games to approximate the original game; (ii) then, we use the previous results to approximate the compact-state games with finite-state models; and (iii) finally, we prove the convergence of the finite-state models to the original model. Notably, steps (ii) and (iii) will be accomplished simultaneously. We impose the assumptions below on the components of the stochastic game. With the exception of the local compactness of the state space, these are the same with Assumption~\ref{compact:as1}.

	\begin{assumption}
		\label{noncompact_as1}
		\begin{itemize}
			\item [ ]
			\item [(a)] The one-stage cost functions $c^i$ are in $C_b(\sX \times \boldsymbol \sA)$.
			\item [(b)] The stochastic kernel $p(\,\cdot\,|x,\boldsymbol a)$ is setwise continuous in $(x,\boldsymbol a)$.
			\item [(c)] $\sX$ is locally compact and $\{\sA_i\}_{i=1}^N$ are compact.
		\end{itemize}
	\end{assumption}

\subsection{Compact Approximation of $N$-player Game} \label{compact_model}

Since $\sX$ is locally compact separable metric space, there exists a nested sequence of compact sets $\{K_n\}$ such that $K_n \subset \intr K_{n+1}$ and $\sX = \bigcup_{n=1}^{\infty} K_n$. Moreover, for any compact $K \subset \sX$, there exists $K_n$ such that $K \subset K_n$ \cite[Corollary 2.77, p. 58]{AlBo06}. Let $\{\nu_n\}$ be a sequence of probability measures such that $\nu_n \in \P(K_n^c)$ for all $n \geq 1$. Similar to the finite-state game construction in Section~\ref{finite_model}, we define a sequence of compact-state games to approximate the original model.

To this end, for each $n$, let $\sX_n = K_n \cup \{\Delta_n\}$, where $\Delta_n \in K_n^c$ is a so-called pseudo-state. We define the transition probability $p_n$ on $\sX_n$ given $\sX_n\times \boldsymbol \sA$ and the one-stage cost functions $c_n^i: \sX_n\times \boldsymbol \sA \rightarrow \R$ by
\begin{align*}
	p_n(\,\cdot\,|x,a) &= \begin{cases}
		p\bigl(\,\cdot\, \bigcap K_n |x,\boldsymbol a\bigr) + p\bigl(K_n^c|x,\boldsymbol a\bigr) \, \delta_{\Delta_n},   &\text{ if } x\in K_n  \\
		\int_{K_n^c} \bigl( p\bigl(\,\cdot\, \bigcap K_n |z,\boldsymbol a\bigr) + p\bigl(K_n^c|z,\boldsymbol a\bigr) \, \delta_{\Delta_n} \bigr) \, \nu_n(dz) ,  &\text{ if } x=\Delta_n,
	\end{cases} \nonumber \\
	\\
	c_n^i(x,\boldsymbol a) &= \begin{cases}
		c^i(x,\boldsymbol a),   &\text{ if } x\in K_n  \\
		\int_{K_n^c} c^i(z,\boldsymbol a) \, \nu_n(dz) ,  &\text{ if } x=\Delta_n. \nonumber
	\end{cases}
\end{align*}
With these definitions, compact-state non-zero sum stochastic game is defined as a stochastic game with the components $\bigl( \sX_n, \boldsymbol \sA, p_n, c_n^1,\ldots,c_n^N\bigr)$. History spaces, policies and cost functions are defined in a similar way as in the original model. To distinguish them from the original game model, we add $n$ as a subscript in each object for the compact model.

Note that under Assumption~\ref{noncompact_as1}, for each $n \geq 1$, compact-state game model with state space $\sX_n$ satisfies Assumption~\ref{compact:as1}. Hence, approximation results established in the previous sections are applicable to this game model. In the rest of this section, we will concentrate on the discounted cost criterion. However, a similar analysis can be applied to the finite-horizon cost criterion under the same set of assumptions. To avoid repetition, we will not include that analysis here.

For each \( n \geq 1 \) and for any initial point $x \in \sX_n$, Theorem~\ref{discounted_thm_2} guarantees the existence of a stationary Nash \(\varepsilon(n)\)-equilibrium \(\boldsymbol{\pi}^*_{n}\) for a compact-state game with state space \(\sX_n\), derived from some finite game model, where \(\varepsilon(n) \rightarrow 0\) as \( n \rightarrow \infty \). Hence, we have the following:
\begin{align*}
	| J_n^i(\boldsymbol \pi^{*}_{n},x) -  J_n^*(\boldsymbol \pi^{*,-i}_{n};x) | \leq \varepsilon(n) \,\,\,\, \forall i=1,\ldots,N.
\end{align*}
Note that, for all $i=1,\ldots,N$, we have 
\begin{align*}
	T^{\boldsymbol \pi^{*,-i}_{n}}_{n} J_{n}^*(\boldsymbol \pi^{*,-i}_{n}; \cdot) = J_{n}^*(\boldsymbol \pi^{*,-i}_{n}; \cdot), 
\end{align*}
where the operator $T^{\boldsymbol \pi^{*,-i}_{n}}_{n}: B(\sX_{n}) \rightarrow B(\sX_{n})$ is defined as 
$$
T^{\boldsymbol \pi^{*,-i}_{n}}_{n} J(x) \coloneqq \min_{a^i \in \sA^i}  \left[ c_n^i(x,\boldsymbol \pi^{*,-i}_{n}(x),a^i) + \beta \, \int_{\sX_n}  J(y) \, p_n(dy|x,\boldsymbol \pi^{*,-i}_{n}(x),a^i) \right].
$$
We now extend the definition of the operator $T^{\boldsymbol \pi^{*,-i}_{n}}_{n}$ to $B(\sX)$ as follows:
$$
\hat T^{\boldsymbol \pi^{*,-i}_{n}}_{n} J(x) \coloneqq \min_{a^i \in \sA^i} \left[ \hat c_n^i(x,\boldsymbol \pi^{*,-i}_{n}(x),a^i) + \beta \, \int_{\sX} J(y) \, \hat p_n(dy|x,\boldsymbol \pi^{*,-i}_{n}(x),a^i) \right], 
$$
where
\begin{align*}
	\hat p_n(\,\cdot\,|x,a) &= \begin{cases} 
		p(\,\cdot\,|x,\boldsymbol a),   &\text{ if } x\in K_n  \\
		\int_{K_n^c} p\bigl(\,\cdot\,|z,\boldsymbol a) \, \nu_n(dz) ,  &\text{ if } x \in K_n^c,
	\end{cases} \\ \\
	\hat c_n^i(x,\boldsymbol a) &= \begin{cases}
		c^i(x,\boldsymbol a),   &\text{ if } x\in K_n  \\
		\int_{K_n^c} c^i(z,\boldsymbol a) \, \nu_n(dz) ,  &\text{ if } x \in K_n^c. 
	\end{cases}
\end{align*}
One can prove that 
\begin{align*}
	\hat T^{\boldsymbol \pi^{*,-i}_{n}}_{n} \hat J_{n}^*(\boldsymbol \pi^{*,-i}_{n}; \cdot) = \hat J_{n}^*(\boldsymbol \pi^{*,-i}_{n}; \cdot), 
\end{align*}
where, in this case, $"$$\,\,\widehat{ }\,\,\,$$"$ means extensions of functions defined on $\sX_{n}$ to $\sX$ as follows:
$$
\hat J(x) = J(x) \,\, \text{if} \,\, x \in K_n, \,\, \hat J(x) = J(\Delta_n) \,\, \text{if} \,\, x \in K_n^c.
$$ 
We can also extend policies in a similar manner, but to avoid complicating the notation, we will not use the notation $"$$\,\,\widehat{ }\,\,\,$$"$ in this case.

\begin{theorem}\label{noncompact_discounted_thm_1}
	Under Assumption~\ref{noncompact_as1}, for each $i=1,\ldots,N$, we have 
	$$
	\lim_{n \rightarrow \infty} \|\hat J_{n}^*(\boldsymbol \pi^{*,-i}_{n}; \cdot)-J^*(\boldsymbol \pi^{*,-i}_{n}; \cdot)\|_K = 0
	$$
	for any compact $K \subset \sX$, where $\|\cdot\|_K$ is the sup-norm on the set $K$.
\end{theorem}

\begin{proof}

To prove the result, we use contraction property of the operators $\hat T^{\boldsymbol \pi^{*,-i}_{n}}_{n}$ and $T^{\boldsymbol \pi^{*,-i}_{n}}$. Indeed, by Banach fixed point theorem, if we start with a common initial function $J_0 \in B(\sX)$, then 
\begin{align*}
	\left(\hat T^{\boldsymbol \pi^{*,-i}_{n}}_{n}\right)^t J_0 &\eqqcolon  \hat J_{n,t}^*(\boldsymbol \pi^{*,-i}_{n}; \cdot) \rightarrow \hat J_{n}^*(\boldsymbol \pi^{*,-i}_{n}; \cdot), \\
	\left( T^{\boldsymbol \pi^{*,-i}_{n}}\right)^t J_0 &\eqqcolon  J_{t}^*(\boldsymbol \pi^{*,-i}_{n}; \cdot) \rightarrow  J^*(\boldsymbol \pi^{*,-i}_{n}; \cdot)
\end{align*}
in sup-norm as $t \rightarrow \infty$ (and so, in sup-norm on any compact set $K$ as $t \rightarrow \infty$). Hence, by using induction, we first prove that 
$$
\lim_{n \rightarrow \infty} \|\hat J_{n,t}^*(\boldsymbol \pi^{*,-i}_{n}; \cdot)-J_t^*(\boldsymbol \pi^{*,-i}_{n}; \cdot)\|_K = 0
$$
for all $t\geq0$ and for any compact $K \subset \sX$. Then, the result follows from the triangle inequality.

Since the function $J_0$ at time zero is common, the result trivially holds. Suppose that the statement is true for $t$ and consider $t+1$. Fix any compact $K \subset \sX$. By definition of $\hat p_n$ and $\hat c_n^i$, there exists $n_0 \geq 1$ such that for all $n \geq n_0$,  we have $\hat p_n = p$ and $\hat c_n^i = c^i$ on $K$ as $K \subset K_n$. With this observation, for each $n \geq n_0$, we have 
\begin{align}
	&\|\hat J_{n,t+1}^*(\boldsymbol \pi^{*,-i}_{n}; \cdot)-J_{t+1}^*(\boldsymbol \pi^{*,-i}_{n}; \cdot)\|_K \\
	&= \sup_{x \in K} \bigg| \min_{a^i \in \sA^i} \biggl[ c^i(x,\boldsymbol \pi^{*,-i}_{n}(x),a^i) + \beta \, \int_{\sX} \hat J_{n,t}^*(\boldsymbol \pi^{*,-i}_{n}; y) \, p(dy|x,\boldsymbol \pi^{*,-i}_{n}(x),a^i) \biggr] \nonumber  \\
	&\phantom{xxx}- \min_{a^i \in \sA^i} \biggl[ c^i(x,\boldsymbol \pi^{*,-i}_{n}(x),a^i) + \beta \, \int_{\sX} J_{t}^*(\boldsymbol \pi^{*,-i}_{n}; dy) \, p(dy|x,\boldsymbol \pi^{*,-i}_{n}(x),a^i) \biggr] \bigg| \nonumber \\
	&\leq \beta  \sup_{(x,a^i) \in K \times \sA^i} \bigg| \int_{\sX} \hat J_{n,t}^*(\boldsymbol \pi^{*,-i}_{n}; y) \, p(dy|x,\boldsymbol \pi^{*,-i}_{n}(x),a^i) \nonumber \\
	&\phantom{xxxxxxxxxxxxxxxxxxxx} -  \int_{\sX} J_{t}^*(\boldsymbol \pi^{*,-i}_{n}; y) \, p(dy|x,\boldsymbol \pi^{*,-i}_{n}(x),a^i) \bigg|. \label{noncompact-eqq1}
\end{align}

\noindent Note that since $p$ is setwise continuous, it is also weakly continuous. Therefore, the set of probability measures $\{p(\cdot|x,\boldsymbol \pi^{*,-i}_{n}(x),a^i)\}_{(x,n,a^i) \in K \times \N \times \sA^i}$ is tight. Hence, for any $\epsilon > 0$, there exists a compact set $K_{\epsilon} \subset \sX$ such that
$$
\sup_{(x,n, a^i) \in K \times \N \times \sA^i} p(K_{\epsilon}^c|x,\boldsymbol \pi^{*,-i}_{n}(x),a^i) \leq \epsilon. 
$$
Let $M \coloneqq \sup_{i=1,\ldots,N} \|c^i\|$. One can prove that 
$$
\|\hat J_{n,t}^*(\boldsymbol \pi^{*,-i}_{n}; \cdot)\|, \, \|J_{t}^*(\boldsymbol \pi^{*,-i}_{n}; \cdot)\| \leq \frac{M}{1-\beta}. 
$$
In view of this, we can obtain
\begin{align*}
	(\ref{noncompact-eqq1}) \leq \beta \, \|\hat J_{n,t}^*(\boldsymbol \pi^{*,-i}_{n}; \cdot)-J_t^*(\boldsymbol \pi^{*,-i}_{n}; \cdot)\|_{K_{\epsilon}} + \beta \, \frac{2 M}{1-\beta} \, \epsilon.
\end{align*}
The first term in the last expression converges to zero as $n \rightarrow \infty$ by the induction hypothesis. Since $\epsilon$ is arbitrary, this completes the proof. 
\end{proof}

We will now establish the central result of this section, which implies that the stationary Nash equilibrium derived from the finite model, when applied to the original model, acts as an approximate stationary equilibrium for the original game.

\begin{theorem}\label{noncompact_discounted_thm_2}
	Under Assumption~\ref{noncompact_as1}, for each $i=1,\ldots,N$, we have 
	$$
	\lim_{n \rightarrow \infty} |J^i(\boldsymbol \pi^{*}_{n}, z)-J^*(\boldsymbol \pi^{*,-i}_{n}; z)| = 0, \,\, \forall z \in \sX,
	$$
where $J^i(\boldsymbol \pi^{*}_{n}, \cdot)$ is the cost of player~$i$ under the joint policy $\boldsymbol \pi^{*}_{n}$. 
\end{theorem}

\begin{proof}

We prove the result by contraction property of the operators. We first define the following operator on $B(\sX)$:
$$
T^{\boldsymbol \pi^{*}_{n},i} J(x) \coloneqq c^i(x,\boldsymbol \pi^{*}_{n}(x)) + \beta \, \int_{\sX} J(y) \, p(dy|x,\boldsymbol \pi^{*}_{n}(x)). 
$$
It is trivial to prove that $T^{\boldsymbol \pi^{*}_{n},i}$ is $\beta$-contraction and the unique fixed point of it is $J^i(\boldsymbol \pi^{*}_{n}, \cdot)$. We also define the following operator on $B(\sX)$:
$$
\hat T^{\boldsymbol \pi^{*}_{n},i} J(x) \coloneqq \hat c_n^i(x,\boldsymbol \pi^{*}_{n}(x)) + \beta \, \int_{\sX} J(y) \, \hat p_n(dy|x,\boldsymbol \pi^{*}_{n}(x)). 
$$
This operator is also $\beta$-contraction and the unique fixed point of it is $\hat J_{n}^i(\boldsymbol \pi^{*}_{n}, \cdot)$, which is the extension of the cost $J_{n}^i(\boldsymbol \pi^{*}_{n}, \cdot)$ of player~$i$ in the compact-state game to the whole state space $\sX$. Using exactly the same arguments that we used in the proof of Theorem~\ref{noncompact_discounted_thm_1}, we can establish that
$$
\lim_{n \rightarrow \infty} \|J^i(\boldsymbol \pi^{*}_{n}, \cdot) - \hat J_{n}^i(\boldsymbol \pi^{*}_{n}, \cdot)\|_K = 0
$$
for any compact $K \subset \sX$. Note that we can always find a compact set $K$ such that $z \in K$, where recall that $z$ is the initial point.

With these observations, we then have 
\begin{align*}
	&|J^i(\boldsymbol \pi^{*}_{n}, z)-J^*(\boldsymbol \pi^{*,-i}_{n}; z)| \\
	&\leq \|J^i(\boldsymbol \pi^{*}_{n}, \cdot) - \hat J_{n}^i(\boldsymbol \pi^{*}_{n}, \cdot)\|_K + | \hat J_{n}^i(\boldsymbol \pi^{*}_{n}, z) - \hat J_n^*(\boldsymbol \pi^{*,-i}_{n};z) | \\
	&\phantom{xxxxxx} + \| \hat J_n^*(\boldsymbol \pi^{*,-i}_{n};\cdot) - J^*(\boldsymbol \pi^{*,-i}_{n}; \cdot)\|_K  \\
	&\leq \|J^i(\boldsymbol \pi^{*}_{n}, \cdot) - \hat J_{n}^i(\boldsymbol \pi^{*}_{n}, \cdot)\|_K + \varepsilon(n) + \| \hat J_n^*(\boldsymbol \pi^{*,-i}_{n};\cdot) - J^*(\boldsymbol \pi^{*,-i}_{n}; \cdot)\|_K  
\end{align*}

\noindent The third term in the last expression converges to zero as $n \rightarrow \infty$ by Theorem~\ref{noncompact_discounted_thm_1}. The first term converges to zero by above argument. By assumption $\varepsilon(n) \rightarrow 0$ as $n \rightarrow \infty$ as well. This completes the proof.  
\end{proof}
Note that for each $i=1,\ldots,N$, we have 
$$ J^*(\boldsymbol \pi^{*,-i}_{n}; x) = \inf_{\pi \in \Pi^i} \cE^{(\boldsymbol \pi^{*,-i}_{n},\pi^i)}_x \left[ \sum_{t=0}^{\infty} \beta^t \, c^i(x_t,\boldsymbol a_t) \right].$$ Hence, Theorem~\ref{noncompact_discounted_thm_2} implies that for any $x \in \sX$ and for any $\varepsilon > 0$, there exists $n(\varepsilon,x)$ such that for any $n \geq n(\varepsilon,x)$, the policy $\boldsymbol \pi^{*}_{n}$ is stationary Nash $\varepsilon$-equilibrium for the initial point $x \in \sX$.

\section{Zero-Sum Markov Games and Markov Teams}\label{zerosum}

In this section, we address the approximation of zero-sum discounted Markov games and discounted Markov teams. As highlighted in the introduction, establishing the existence of stationary equilibria for zero-sum discounted Markov games is significantly simpler than in the nonzero-sum setting \cite{Sha53,MaPa70,Cou80,Now03}, \cite[Chapter 5]{BaZa18}. Additionally, since the methodology for proving the existence of stationary equilibria in zero-sum settings closely parallels that of MDPs through the use of optimality operators, finite approximations for such problems can also be achieved under mild conditions \cite{TiAl96,TiMaPoAl97,PeScPiPi15}, similar to the MDP setting \cite{SaYuLi15c,SaLiYuSpringer}. These cases highlight the distinctions between the zero-sum and nonzero-sum settings and demonstrate why the analysis of nonzero-sum settings is significantly more challenging. 

\subsection{Zero-Sum Markov Games}\label{zerosum1}

In the zero-sum setting, the gain of one player equals the loss of the other. Consequently, at equilibrium, the cost function achieved by one player is the negative of the cost function of the other. This symmetry ensures that there is a single cost, referred to as the \emph{value of the game}, that both players are concerned with. This value is unique under general conditions. Moreover, under mild regularity assumptions, it can be shown that the value function is a fixed point of a min-max contraction operator. The existence of stationary equilibria then follows from the min-max measurable selection theorem. This solution strategy closely mirrors the methods used in MDPs. Given that the analysis presented in this paper relies heavily on recent developments in \cite{SaYuLi15c,SaLiYuSpringer} about the approximation of MDPs, the results regarding the approximation of zero-sum discounted Markov games can also be established using the same proof techniques as those in \cite{SaYuLi15c,SaLiYuSpringer}. For this reason, the result stated below for zero-sum games is given without proof.

	\begin{assumption}
		\label{zero-sum_as1}
		\begin{itemize}
			\item [ ]
			\item [(a)] The one-stage cost function $c$ is in $C_b(\sX \times \sA^1 \times \sA^2)$.
			\item [(b)] The stochastic kernel $p(\,\cdot\,|x,a^1,a^2)$ is weakly continuous in $(x,a^1,a^2)$.
			\item [(c)] $\sX$ is locally compact and $\{\sA^i\}_{i=1}^2$ are compact.
		\end{itemize}
	\end{assumption}

In this setting, the gain of the one player is the loss of the other player:
\begin{align*}
	J^1(\pi_1,\pi_2,x) &= \cE_{x}^{\pi_1,\pi_2}\biggl[\sum_{t=0}^{\infty}\beta^{t} \, c(x_{t},a_t^1,a_t^2)\biggr] = -J^2(\pi_1,\pi_2,x). \nonumber
\end{align*}
Hence, $J^*(x) := \min_{\pi_1} \max_{\pi_2} J^1(\pi_1,\pi_2,x)$ (i.e. the value of the game) is the unique fixed point of the contraction operator $T:C_b(\sX) \rightarrow C_b(\sX)$:
\begin{align*}
TJ(x) &\coloneqq \min_{\mu \in \P(\sA^1)} \max_{\nu \in \P(\sA^2)} \int_{\sA^1\times\sA2} \biggl[c(x,a^1,a^2) \\
&\phantom{xxxxxxxxxxxxxxxxxxxx} + \beta \int_{\sX} J(y) \, p(dy|x,a^1,a^2) \biggr] \nu(da^1) \, \mu(da^2). 
\end{align*}
Once the value of the game is determined, the stationary equilibrium \((\pi_1^*, \pi_2^*)\) corresponds to the min-max measurable selectors of the above optimality equation when \(J\) is replaced by \(J^*\). In fact, it can also be shown that these optimal selectors are deterministic functions. This analysis is essentially analogous to the MDP setting (see \cite{HeLa96}), which is why the results from \cite{SaYuLi15c,SaLiYuSpringer}-originally developed for the approximation of MDPs-can be applied with only minor modifications. As a result, we can obtain the following theorem. 

\begin{theorem}\label{zerosum_thm}
	\begin{itemize}
		\item[ ]
		\item \textbf{Compact State}: Under Assumption~\ref{zero-sum_as1} and compactness of the state space, we have 
		$$
		\lim_{\delta \rightarrow 0} \|J^1(\pi^{1,*}_{\delta},\pi^{2,*}_{\delta}, \cdot)-J^*(\cdot)\| = 0, 
		$$
		where $J^1(\pi^{1,*}_{\delta},\pi^{2,*}_{\delta}, \cdot, \cdot)$ is the cost of player~$1$ under the equilibrium $(\pi^{1,*}_{\delta},\pi^{2,*}_{\delta})$ of the $\delta$-approximate game. 
		\item \textbf{Non-Compact State}: Under Assumption~\ref{zero-sum_as1}, we have 
		$$
		\lim_{n \rightarrow \infty} \|J^1(\pi^{1,*}_{n}, \pi^{2,*}_{n}, \cdot)-J^*(\cdot)\|_K = 0
		$$
		for any compact $K \subset \sX$, where $J^1(\pi^{1,*}_{n}, \pi^{2,*}_{n}, \cdot)$ is the cost of player~$1$ under the equilibrium $(\pi^{1,*}_{n}, \pi^{2,*}_{n})$ of the $n^{th}$-compact-state approximation. 
	\end{itemize}
\end{theorem}

\begin{proof}[Proof Sketch]
The min-max Bellman operator is a contraction on \(C_b(\mathbb{X})\), and hence the value function of the original Markov game belongs to \(C_b(\mathbb{X})\). This enables a more direct proof: the continuous-space extension of the value functions of the finite-model games converge uniformly to the value function of the original game as the quantization becomes finer. We note that this uniform convergence property does not necessarily hold in the non-zero-sum setting. This also implies continuous convergence of the extended value functions, allowing the desired result to be established by a direct application of \cite[Lemmas 2.2 and 2.3]{SaYuLi15c}, together with an appropriate weak version of the generalized dominated convergence theorem of \cite[Theorem 3.5]{Lan81} or \cite[Theorem 3.5]{serfozo1982convergence}.
\end{proof}

As can be seen from the theorem above, in the zero-sum setting, we can relax the setwise continuity condition required in the non-zero-sum setting to the weaker condition of weak continuity. Moreover, in the zero-sum case, the approximation result holds uniformly over the initial state, whereas in the non-zero-sum case, the result is only pointwise in the initial state. We note again that the value sequence of the finite games also converges to the continuous limit. Thus, the zero-sum setting allows us to obtain a stronger approximation result under weaker continuity assumptions than those required in the non-zero-sum setting. The theorem above can also be interpreted in the context of continuity properties of value functions under weakly converging de-garbled information structures \cite[Theorem 8.3.4, Section 11.4]{YukselBasarBook24}.

\subsection{Markov Teams}

In the team setting, the players are minimizing the same cost function:
\begin{align*}
	J^i(\boldsymbol \pi,x) &= \cE_{x}^{\boldsymbol \pi}\biggl[\sum_{t=0}^{\infty}\beta^{t} \, c(x_{t},\boldsymbol a_t)\biggr] = J^j(\boldsymbol \pi,x), \,\, \forall i,j. \nonumber
\end{align*}
Hence, this problem is indeed equivalent to the MDP problem with the following components: $(\sX,\boldsymbol \sA, p, c)$. In this setting, as in the previous subsection, we impose the following conditions.

	\begin{assumption}
		\label{team_as1}
		\begin{itemize}
			\item [ ]
			\item [(a)] The one-stage cost function $c$ is in $C_b(\sX \times \boldsymbol \sA)$.
			\item [(b)] The stochastic kernel $p(\,\cdot\,|x,\boldsymbol a)$ is weakly continuous in $(x,\boldsymbol a)$.
			\item [(c)] $\sX$ is locally compact and $\{\sA_i\}_{i=1}^N$ are compact.
		\end{itemize}
	\end{assumption}

In this setup, the optimal value of the team $J^*(x) := \min_{\boldsymbol \pi} J^i(\boldsymbol \pi,x)$ is the unique fixed point of the contraction operator $L:C_b(\sX) \rightarrow C_b(\sX)$:
\small
$$
LJ(x) \coloneqq \min_{\boldsymbol a \in \boldsymbol \sA}  \left[c(x,\boldsymbol a) + \beta \int_{\sX} J(y) \, p(dy|x,\boldsymbol a) \right]. 
$$
\normalsize
Once the optimal value of the game is determined, the stationary optimal joint policy \(\boldsymbol \pi\) corresponds to the measurable selector of the above optimality equation when \(J\) is replaced by \(J^*\). This analysis is exactly the same with the MDP setting (see \cite{HeLa96}), and so, the results from \cite{SaYuLi15c,SaLiYuSpringer} can be applied without any modification. As a result, we can derive the following theorem. 

\begin{theorem}\label{team_thm}
	\begin{itemize}
		\item[ ]
		\item \textbf{Compact State}: Under Assumption~\ref{team_as1} and compactness of the state space, we have 
		$$
		\lim_{\delta \rightarrow 0} \|J^i(\boldsymbol \pi^{*}_{\delta}, \cdot)-J^*(\cdot)\| = 0, 
		$$
		where $J^i(\boldsymbol \pi^{*}_{\delta}, \cdot)$ is the common cost under the optimal joint policy $\boldsymbol \pi^{*}_{\delta}$ of the $\delta$-approxi- mate team. 
		\item \textbf{Non- Compact State}: Under Assumption~\ref{team_as1}, we have 
		$$
		\lim_{n \rightarrow \infty} \|J^i(\boldsymbol \pi^{*}_{n}, \cdot)-J^*(\cdot)\|_K = 0
		$$
		for any compact $K \subset \sX$, where $J^i(\boldsymbol \pi^{*}_{n}, \cdot)$ is the common cost under the optimal joint policy $\boldsymbol \pi^{*}_{n}$ of the $n^{th}$-compact-state approximation. 
	\end{itemize}
\end{theorem}

\section{Conclusion}

In this paper, we have established the existence of near Markov and stationary perfect Nash equilibria for nonzero-sum stochastic games using a finite state-action approximation method, under both finite-horizon and discounted cost criteria, addressing both compact and non-compact state spaces. 



\end{document}